\documentclass[twocolumn,superscriptaddress,showpacs,prx]{revtex4-2}
\usepackage[latin9]{inputenc}
\pdfoutput=1
\usepackage[english]{babel}
\usepackage[T1]{fontenc}
\usepackage{hyperref}
\usepackage{orcidlink}

\usepackage{wrapfig}

\usepackage{amsmath, amssymb, amstext, amsfonts}
\usepackage{amsthm}
\usepackage{float}
\usepackage{graphicx}
\usepackage{xcolor}
\usepackage[shortlabels]{enumitem}
\usepackage{physics}

\def\R{\mathbb{R}}

\def\P{{\cal P}}

\newtheorem{theo}{Theorem}
\newtheorem{remark}{Remark}
\newtheorem{defin}[theo]{Definition}

\newtheorem{lemma}[theo]{Lemma}
\newtheorem{prop}[theo]{Proposition}

\def\id{{\mathbb I}}

\def\bra#1{\langle#1|} \def\ket#1{|#1\rangle}
\def\braket#1#2{\langle#1|#2\rangle}

\def\ketbra#1#2{\ket{#1}\!\bra{#2}}
\def\proj#1{\ket{#1}\!\bra{#1}}

\def\be{\begin{equation}}
\def\ee{\end{equation}}
\def\bea{\begin{eqnarray}}
\def\eea{\end{eqnarray}}
\def\bma{\begin{mathletters}}
\def\ema{\end{mathletters}}

\def\P{{\cal P}}

\def\q0{\underline{0}}

\def\Z{{\cal Z}}

\def\P{{\cal P}}

\def\C{{\mathbb C}}
\def\id{{\mathbb I}}

\def\R{\mathbb{R}}

\begin{document}

\title{Self-testing in prepare-and-measure scenarios and a robust version of Wigner's theorem}
\author{Miguel Navascu\'es}
\affiliation{Institute for Quantum Optics and Quantum Information (IQOQI) Vienna\\ Austrian Academy of Sciences}
\author{K\'aroly F. P\'al}
\affiliation{Institute for Nuclear Research, P.O. Box 51, H-4001 Debrecen, Hungary}
\author{Tam\'as V\'ertesi}

\affiliation{MTA Atomki Lend\"ulet Quantum Correlations Research Group,\\
Institute for Nuclear Research, P.O. Box 51, H-4001 Debrecen, Hungary}
\author{Mateus Ara\'ujo}
\affiliation{Departamento de F\'isica Te\'orica, At\'omica y \'Optica, Universidad de Valladolid, 47011 Valladolid, Spain}

\begin{abstract}
We consider communication scenarios where one party sends quantum states of known dimensionality $D$, prepared with an untrusted apparatus, to another, distant party, who probes them with uncharacterized measurement devices. We prove that, for any ensemble of reference pure quantum states, there exists one such prepare-and-measure scenario and a linear functional $W$ on its observed measurement probabilities, such that $W$ can only be maximized if the preparations coincide with the reference states, modulo a unitary or an anti-unitary transformation. In other words, prepare-and-measure scenarios allow one to `self-test' arbitrary ensembles of pure quantum states. Arbitrary extreme $D$-dimensional quantum measurements, or sets thereof, can be similarly self-tested. Our results rely on a robust generalization of Wigner's theorem, a well-known result in particle physics that characterizes physical symmetries.

\end{abstract}

\maketitle

All experiments in quantum physics start by setting the lab equipment in a given state and end by conducting a measurement. When we assign different experimenters to each of these two tasks, namely, when one experimenter is asked to prepare certain quantum states; and another one, to probe them, then we are working in the so-called prepare-and-measure communication scenario~\cite{Pawlowski2011,Tavakoli2018PM}. This paradigm models many primitives of interest for quantum information theory, such as quantum key distribution (QKD)~\cite{Scarani2009,Woodhead2015}, quantum communication complexity~\cite{Ambainis2002} and metrology~\cite{Giovannetti2006}.

A communication protocol that does not rely on a characterization of the measurement and preparation apparatuses is said to be \emph{device-independent} (DI)~\cite{Acin2007}. The security or success of those protocols is thus guaranteed by their measurement statistics alone~\cite{colbeck2009,Pironio2010}. Unfortunately, prepare-and-measure scenarios cannot be fully DI, as an arbitrarily large classical memory suffices to explain all conceivable measurement statistics. 

It is possible, however, to devise \emph{semi-device independent} (SDI) prepare-and-measure protocols~\cite{Pawlowski2011}. SDI protocols rely both on the measurement statistics and also on some (generally, weak) promise on the preparation or measurement devices~\cite{VanHimbeeck2017energy}, \cite{Tavakoli2020info},\cite{jones2022theoryindependent}. Following most of the literature on SDI protocols (see, e.g., \cite{Gallego2010,Pawlowski2011,Woodhead2015}), in this Letter we will posit a bound on the Hilbert space dimension of the preparations. Note that the SDI paradigm allows certifying properties that the DI approach cannot, e.g.: self-testing non-projective measurements \cite{Tavakoli2020POVM}. In addition, SDI protocols are in general experimentally friendlier than their DI counterparts. 

Under the assumption that the Hilbert space dimension of the prepared systems is known, it was observed in Ref.~\cite{Tavakoli2018PM,Tavakoli2020POVM} that certain qubit states and qubit measurements could be \emph{self-tested}: namely, the only way to generate certain feasible measurement statistics in a prepare-and-measure scenario is to prepare those quantum states and conduct those measurements, modulo unitary/anti-unitary transformations. In the same spirit, the authors of \cite{Mate} show how to self-test measurements of Mutually Unbiased Bases (MUBs) \cite{MUBs} in arbitrary dimensions. These works leave us with the question of which state ensembles and measurements, in the qubit case, as well as in higher dimensions, can be self-tested.

In this paper, we answer this question, by providing a family of linear witnesses whose maximal value self-tests arbitrary ensembles of pure states and arbitrary sets of extreme POVMs (or both), in prepare-and-measure scenarios of arbitrary Hilbert space dimension. Since neither state ensembles containing mixed states nor non-extreme Positive Operator Valued Measures (POVMs) can be self-tested, our result fully characterizes the limits of self-testing in the prepare-and-measure scenario. Note that, prior to our work, general self-testing schemes only existed for scenarios with non-demolition measurements, which should be sequentially applied in the course of a single experimental round \cite{huang2022foundations}, \cite{Das_2022}.

To prove our main result, we generalize the famous Wigner's theorem~\cite{wigner1931}, which states that all physical symmetries (i.e., all maps from rays to rays that preserve the absolute value of the scalar product) can be expressed as a unitary or an anti-unitary transformation. Our generalization considers ``noisy partial symmetries'', whose domain is limited to a finite number of rays and which preserve the absolute value of the scalar product up to an error. 



The scenario we consider differs from that of other works with a similar flavor, such as Miklin and Oszmaniec's \cite{Miklin_2021}. In this paper, the authors assume that, every time that the experiment is reset, the same state preparations and measurements are available -such scenarios are usually called independent and identically distributed (i.i.d.). To the contrary, the statistical tests proposed in this Letter do not rely on the i.i.d. assumption: our results are therefore robust under the miscalibration of the preparation and measurement devices and even allow for correlations between the different experimental rounds. 


\vspace{10pt}
\noindent\emph{Prepare-and-measure scenarios}
\vspace{10pt}

In a prepare-and-measure scenario, one party, Alice, prepares a $D$-dimensional quantum state labeled by the index $x=1,...,X$, and sends it to a second party, Bob, who probes the state with some POVM $y\in\{1,...,Y\}$, obtaining an outcome $b\in\{1,...,d\}$. The scenario is thus specified by the vector of natural numbers $(D, X, Y, B)$. 

In the following, we denote Alice's $x^{th}$ state as $\bar{\psi}_x\in B(\C^D)$; and Bob's $y^{th}$ POVM, by $\bar{M}_y:=(\bar{M}_{b|y}\in B(\C^D):b=1,...,B)$. We will call $\bar{\psi}$ ($\bar{M}$) Alice's collection of states (Bob's collection of POVMs). That is, $\bar{\psi}=\{\bar{\psi}_x\}_{x=1}^X$, $\bar{M}:=\{\bar{M}_y\}_{y=1}^Y$. Of course, both states and POVMs are subject to the usual positivity and normalization conditions, i.e., $\bar{\psi}_x\geq 0,\tr(\bar{\psi}_x)=1,\forall x$, $\bar{M}_{b|y}\geq 0,\forall y,b$ and $\sum_b\bar{M}_{b|y}=\id_D, \forall y$.

Denoting by $P(b|x,y)$ the probability that Bob observes outcome $b$ when he performs measurement $y$ on state $x$, the experiment's measurement statistics $P:=(P(b|x,y):x,y,b)$ are given by
\begin{equation}
P(b|x,y)=\tr(\bar{\psi}_x \bar{M}_{b|y}),\forall x,y,b.
\label{def_prob}
\end{equation}
 Note that, for any unitary or antiunitary $U$, the state ensemble $U\psi U^\dagger$ and the measurements $UMU^\dagger$ generate the same probability distribution $P(b|x,y)$ as the original state ensemble $\psi$ and measurement set $M$ used by Alice and Bob. The realizations $(U \psi U^\dagger,U M U^\dagger)$, $(\psi,M)$ are therefore operationally indistinguishable within the semi-device independent paradigm. We call $Q_D$ the set of all distributions $P$ admitting some $D$-dimensional realization $(\psi, M)$.

Consider a prepare-and-measure scenario $(D, X,Y,B)$, and let $W:\R^{XYB}\to\R$ be a linear functional with $\max_{P\in Q_D}W(P)=W^\star$. For ${\cal X}\leq X$, ${\cal Y}\leq Y$, we say that $W$ \emph{self-tests} the states $(\psi_x)_{x=1}^{{\cal X}}$ and the POVMs $\{M_y\}_{y=1}^{{\cal Y}}$ if, for any feasible $P$ realized by $(\bar{\psi}, \bar{M})$ with $W(P)=W^\star$, there exists a unitary or anti-unitary map $U$ with the property that
\begin{align}
&\bar{\psi}_x=U\psi_xU^\dagger,x=1,...,{\cal X},\nonumber\\
&\bar{M}_{b|y}=U M_{b|y}U^\dagger,y=1,...,{\cal Y},b=1,...,B.
\label{self_test}
\end{align}
We say that $W$ \emph{robustly self-tests} $(\psi_x)_{x=1}^{{\cal X}}$, $\{M_y\}_{y=1}^{{\cal Y}}$ if, for all $\epsilon>0$, there exists $\epsilon'>0$ such that $W^\star- W(P)\leq \epsilon'$ implies that relations (\ref{self_test}) are satisfied up to precision $\epsilon$ in trace and operator norm, respectively.

No linear functional $W$ can self-test non-extreme states or measurements. Suppose, e.g., that $W$ were maximized by a feasible distribution $P$ whose realization involved a mixed state $\bar{\psi}_x=\sum_{j}\lambda_j\proj{\bar{\phi}_j}$, with $\lambda_j>0$. Then, the distribution $P'$ generated if we replaced $\bar{\psi}_x$ with $\proj{\bar{\phi}_1}$ would also maximize $W$. However, $\proj{\bar{\phi}_1}$ and $\bar{\psi}_x$ are not connected by a unitary or an anti-unitary transformation. The same argument holds for extreme POVMs. We arrive at the conclusion that only extreme (pure) states and extreme POVMs can be, in principle, self-tested, modulo unitary or anti-unitary transformations.

To prove that a prepare-and-measure experimental system satisfies an inequality of the form $W(P)\geq W^\star-\epsilon$, one would first think of estimating the probabilities $P(b|x,y)$ through repeated experiments and then evaluating the witness. However, such a direct approach is only feasible when the experimental setup satisfies the i.i.d. assumption. When the system is not i.i.d., the goal is to reject the null hypothesis that, in each experimental round, $W(P)< W^\star-\epsilon$. It turns out that, as long as the functional $W$ is linear, it is possible to devise a statistical test that fits the bill and yet does not rely on the i.i.d. assumption \cite{zhang}. In this test, at each experimental round the inputs $x,y$ are sampled according to some probability distribution and the output $b$ is used to generate a round score. The score of the different rounds is multiplied and a $p$-value for the null hypothesis is derived. If the system violates the hypothesis and is approximately i.i.d., the $p$ value will quickly tend to zero as the number of rounds increases \cite{zhang}.

Ref.~\cite{Tavakoli2020POVM} proposes linear witnesses to self-test some extreme POVMs and state ensembles in the qubit case ($D=2$). The goal of the rest of the paper is to generalize the results of Ref.~\cite{Tavakoli2020POVM} to robustly self-test any ensemble of pure states and collection of extreme POVMs, defined in Hilbert spaces of arbitrary dimension $D$. 

\vspace{10pt}

\noindent\emph{Self-testing of pure state ensembles}

\vspace{10pt}

We start by showing how to self-test pure state ensembles. To avoid a cumbersome notation, from now on, whenever we refer to a normalized ket $\ket{\omega}$, we will use $\omega$ to denote its corresponding rank-$1$ projector $\proj{\omega}$. For fixed dimension $D$, we call $\P\subset B(\C^D)$ the set of all rank-$1$ projectors. 

Self-testing of state ensembles is based on the following lemma.

\begin{lemma}
\label{witness_norms}
Let $\Psi\equiv\{\psi_i\}_{i=1}^N\subset \P$ be a collection of pure quantum states, such that
\begin{equation}
\sum_i\alpha_i\psi_i=\frac{\id_D}{D},
\label{max_mixed_cond}
\end{equation}
for some $\{\alpha_i\}_{i=1}^N\subset \R^+$. Consider a prepare-and-measure scenario $(D, N,\frac{N^2-N}{2},2)$, with the measurements labeled by $y\in\{(i,j):i>j,i,j=1,...,N\}$. Define the linear witness:
\begin{align}
W_\Psi(P):=&\sum_{i>j}\alpha_i\alpha_j\|\psi_i-\psi_j\|_1S_{i,j},
\end{align}
with
\begin{equation}
S_{i,j}:= P(2|x=i,y=(i,j))-P(2|x=j,y=(i,j)).
\end{equation}
Then, for all $P\in Q_D$, it holds that
\begin{equation}
W_\Psi(P)\leq 1-\frac{1}{D}.
\end{equation}
This inequality is tight, and can be saturated by preparing the states $\Psi$ and choosing the dichotomic measurements appropriately.

Moreover, if any feasible distribution $P$, realized with preparation states $\{\bar{\psi}_i\}_{i=1}^N\subset B(\C^D)$, satisfies $W_\Psi(P)\geq 1-\frac{1}{D}-\epsilon$, then it holds that 
\begin{align}
&|\tr\{\bar{\psi}_i\bar{\psi}_j\}-\tr\{\psi_i\psi_j\}|\leq O\left(\sqrt{\epsilon}\right),\forall i\neq j.  \nonumber \\
&1-\tr(\bar{\psi}_i^2)\leq O(\epsilon),\forall i
\label{basic_relation}
\end{align}
In particular, when $\epsilon=0$, then all the prepared states $\{\bar{\psi}\}_{i=1}^N$ are pure and have the same projector overlaps as the reference states $\psi$.
\end{lemma}
This lemma can be regarded as a study of the saturation conditions of a variant of the dimension witness proposed in \cite{BNV}. The reader can find a proof in Section \ref{app_witness_norms} of the Appendix, where the exact expressions for the right-hand sides of eq. (\ref{basic_relation}) are provided.

Now, suppose that we wished to self-test an ensemble of pure-state preparations $\{\psi_i\}_{i=1}^M$. To exploit Lemma \ref{witness_norms}, we need to find (pure) states $\{\psi_i\}_{i=M+1}^N$ and positive real numbers $\{\alpha_i\}_{i=1}^N$ such that condition (\ref{max_mixed_cond}) holds. Such extra states and positive numbers always exist: consider, for instance, the maximum $\lambda\in \R$ such that the operator $V=\frac{\id}{D}-\lambda\sum_{i=1}^M\psi_i$ is positive semidefinite \footnote{This maximum is always positive, as $\lambda = 1/(MD)$ will already give a positive semidefinite operator.}. Let $\sum_{i=M+1}^{N}\beta_i\proj{\psi_i}$ be the spectral decomposition of $V$, with $\beta_i>0$ (we omit the eigenvectors with zero eigenvalue, so $N\leq M+D-1$). Then we have that $\{\psi_i\}_{i=1}^{N}$, and $\{\alpha_i\}_{i=1}^{N}$, with $\alpha_i:=\lambda$, for $i=1,\ldots,M$, and $\alpha_i := \beta_i$, for $i=M+1,\ldots,N$, satisfy condition (\ref{max_mixed_cond}). 

Given $\psi=\{\psi_i\}_{i=1}^N$, $\{\alpha_i\}_{i=1}^N$, we can thus build the witness $W_\psi(P)$. By Lemma \ref{witness_norms}, if $W_\psi(P)$ is $\epsilon$-close to its maximum value, then the prepared states $\{\phi_i\}_{i=1}^N$ will satisfy eq. (\ref{basic_relation}). The question is whether, for $\epsilon$ sufficiently small, this condition implies that $\phi_i\approx U\psi_iU^\dagger$, for all $i$, for some unitary or anti-unitary $U$.

Note the similarities with the famous Wigner's theorem \cite{wigner1931,wigner_old,Geher2014}, whose finite-dimensional version reads \footnote{The original version of Wigner's theorem is expressed in terms of Hilbert space rays $\ket{\phi},\ket{\phi'}$ and the complex modulus of their normalized overlaps $o(\phi,\phi')=\frac{|\braket{\phi}{\phi'}|}{\sqrt{\braket{\phi}{\phi}}\braket{\phi'}{\phi'}}$. Note, however, that $\tr(\proj{\tilde{\phi}}\proj{\tilde{\phi'}})=o(\phi,\phi')^2$, with $\ket{\tilde{\phi}}=\frac{\ket{\phi}}{\sqrt{\braket{\phi}{\phi}}}$. It follows that we can express Wigner's theorem in terms of rank-$1$ projectors.}:
\begin{theo}
Let the (possibly non-linear) map $\omega:\P\to\P$ have the property:
\begin{equation}
\tr(\phi \phi')=\tr(\omega(\phi) \omega(\phi')), \forall \phi,\phi'\in \P.
\label{wigner_overlap}
\end{equation}
Then, there exists a unitary or anti-unitary transformation $U$ such that
\begin{equation}
\omega(\phi)=U\phi U^\dagger, \forall \phi\in \P.
\end{equation}
\end{theo}

We wish to generalize this result in two ways. First, in our case the domain of $\omega$ only covers a finite set of rank-$1$ projectors, namely, $\{\psi_i\}_{i=1}^N$. Second, we are interested in situations where eq. (\ref{wigner_overlap}) only holds approximately. This leads us to define what from now on we call the \emph{Wigner property}.

\begin{defin}
\label{def_wigner}
We say that a set of pure states $\{\psi_i\}_{i=1}^N\subset\P$ satisfies the Wigner property if, for all $\delta'>0$ and for any set of (not necessarily pure) states $\{\bar{\psi}_i\}_{i=1}^N\subset B(\C^D)$, there exists $\delta>0$ such that the relation
\begin{equation}
|\tr(\psi_i\psi_j)-\tr(\bar{\psi}_i\bar{\psi}_j)|\leq \delta, \forall i,j
\label{overlaps}
\end{equation}
implies that there exist a(n) (anti-)unitary transformation $U$ with $\|\psi_i-U\bar{\psi}_i U^\dagger\|_1\leq \delta'$, for all $i$.
\end{defin}

As observed in \cite{Miklin_2021}, for $D=2$ all pure state ensembles satisfy the Wigner property. In that case, 
\begin{equation}
\psi_i=\frac{\id+\vec{m}^i\cdot\vec{\sigma}}{2},\bar{\psi}_i=\frac{\id+\bar{n}^i\cdot\vec{\sigma}}{2},
\end{equation}
for some vectors $\{\vec{m}^i,\vec{n}^i:\|\vec{m}^i\|,\|\vec{n}^i\|\leq 1\}_{i=1}^N\subset \R^3$. Here $\vec{\sigma}=(\sigma_x,\sigma_y,\sigma_z)$ denotes the three Pauli matrices. Setting $\delta= 0$ in eq. (\ref{overlaps}), we have that $\vec{m}^i\cdot\vec{m}^j= \vec{n}^i\cdot\vec{n}^j$, for all $i,j$. It follows that there exists an orthogonal transformation $O$ such that $\vec{n}^i= O\vec{m}^i$, for all $i$. Any orthogonal transformation in $\R^3$ can be expressed as either $R$ or $TR$, where $R$ represents a rotation; and $T$, a reflection. In the first case, there exists a unitary $U$ such that $\phi_i=U\psi_iU^\dagger$, for all $i$. In the second case, there exists an anti-unitary operation $V$ such that $\phi_i= V\psi_iV^\dagger$. In Section \ref{app:wigner_2} of the Appendix we provide a robust version of this argument, which proves that any ensemble of pure states in dimension $D=2$ satisfies the Wigner property with $\delta'=O(\delta^{1/4})$.

How about higher dimensions? Do pure state ensembles in, say, dimension $3$, satisfy the Wigner property, too? In general, no. In Section \ref{app:SIC} of the Appendix we present several examples of pairs of state ensembles in dimension $3$ that, despite having the same overlaps, are not related via unitary or anti-unitary transformations. Furthermore, we find that generic ensembles of three two-dimensional pure states, embedded in $B(\C^3)$, do not satisfy the Wigner property, even if we restrict it to the zero-error case ($\delta=0$).

How can we then certify state ensembles of dimensions greater than two? A possible way to solve this problem is to look for inspiration in the recent literature on Wigner's theorem. In this regard, the proof of Wigner's theorem in \cite{Geher2014} relies on the existence of a set of $5D-6$ pure states ${\cal T}\subset \P$ with the following property: for any ensemble of pure states $\Psi=\{\psi_i\}_{i=1}^N\subset \P$ such that $\bra{k}\psi_i\ket{k}>0,\forall k, i$, the overlaps between the states in $\Psi\cup{\cal T}$ uniquely identify this latter set, modulo a(n) (anti-)unitary transformation. In Section \ref{app:robust} of the Appendix we make this statement robust. Namely, we prove the following result. 

\begin{lemma}
Let $\{\psi_i\}_{i=1}^N\subset \P$ be such that $\tr(\psi_i\proj{k})\geq f>0$, for $k=1,...,D$. Then, the ensemble of pure states ${\cal T}\cup\{\psi_i\}_{i=1}^N$ satisfies the Wigner property with $\delta'=O(\sqrt{f}D^{7/4}\delta^{1/8})$.
\label{wigner_discrete}
\end{lemma}
To arrive at robustness bounds that scale well with $D$, the proof makes use of Hausladen and Wootters' pretty good measurement \cite{square_root}, duality theory \cite{sdp} and exactly solvable tridiagonal matrices. The exponent $1/8$ on $\delta$ is admittedly very inconvenient. Presumably, one could achieve better robustness bounds by taking a tomographically complete fiducial set ${\cal T}'$ instead of ${\cal T}$. This is the approach used in \cite{huang2022foundations}, which follows more closely the original proof of Wigner's theorem. The tomographic approach has the disadvantage of requiring $O(D^2)$ new state preparations, instead of $O(D)$.

Now, suppose that we wish to self-test the ensemble of preparations $\psi=\{\psi_i\}_{i=1}^M\subset \P$. First, we transform the computational basis $\{\ket{k}\}_{k=1}^{D}$ with a unitary to ensure that all $M$ states in $\psi$ satisfy $\tr(\psi_i \proj{k})>0,\forall k$ (a random unitary will achieve this with probability $1$). Next, we consider the ensemble of preparations $\tilde{\psi}:=\psi\cup {\cal T}\cup {\cal R}$, where ${\cal R}$ are extra states (not to be self-tested) that we might need to add to ensure that the ensemble $\tilde{\psi}$ satisfies condition (\ref{max_mixed_cond}) for some positive numbers $(\alpha_i)_i$. 

Suppose that the corresponding dimension witness $W_{\tilde{\psi}}(P)$ is maximized by the set of (necessarily pure) states $\bar{\psi}\cup\bar{{\cal T}}\cup\bar{{\cal R}}$. Then condition (\ref{overlaps}) and Lemma \ref{wigner_discrete} guarantee that the ensembles of states $\psi\cup {\cal T}$, $\bar{\psi}\cup \bar{{\cal T}}$, are related by a unitary or an anti-unitary transformation. In particular, the witness $W_{\tilde{\psi}}$ self-tests the reference states $\{\psi_i\}_{i=1}^M$. This result can be made robust by applying Lemmas \ref{witness_norms} and \ref{wigner_discrete} in sequence. Thus, a value of $W_\psi$ that is $\epsilon$-short from maximum indicates that $U\bar{\psi}_iU^\dagger$ is $O(\epsilon^{1/16})$-away from $\psi_i$, for all $i$.

\vspace{10pt}

\noindent\emph{Self-testing of extremal POVMs}

\vspace{10pt}

We now turn to the problem of self-testing extremal POVMs. We will rely on the characterization of extreme quantum measurements by D'Arianno et al. \cite{POVMs_dariano}. Namely: a POVM $(M_b)_b$ is extremal iff, for any set of $D\times D$ Hermitian matrices $(H_b)_b$, with $\mbox{Supp}(H_b)\subset \mbox{Supp}(M_b),\forall b$, the condition $\sum_b H_b=0$ implies that $H_b=0,\forall b$.

Now, let $(M_b)_b$ be an extreme POVM, and, for each $b$, let $Z_b$ be a projector onto the kernel of $M_b$. Then, the only maximizer of the POVM optimization problem $\max_{\bar{M}}-\sum_b\tr(Z_b\bar{M}_b)$ is $M$. Indeed, first note that the maximum value of the problem is zero, which can, indeed, be achieved by the solution $\bar{M}=M$. Now, suppose that there exists another solution $M^\star$ of the optimization problem. Then, $\tr(M^\star_b Z_b)=0$, for all $a$, which implies that $\mbox{Supp}(M^\star_b)\subset \mbox{Kern}(Z_b)=\mbox{Supp}(M_b)$. Define then $H_b:=M^\star_b-M_b$. Then on one hand we have that $\mbox{Supp}(H_b)\subset \mbox{Supp}(M_b)$, for all $b$. On the other hand, $\sum_bH_b=\sum_bM^\star_b-\sum_bM_b=\id-\id=0$. By the extremality of $M$ it thus follows that $H_b=0$, for all $b$, and so $M^\star=M$.

Combined with our tool for self-testing states, this observation is enough to self-test $(M_b)_{b=1}^B$. Let $Z_b$ admit a spectral decomposition as $Z_b=\sum_{i=1}^{d_b}\proj{\psi^a_i}$ and define the pure state ensemble $\psi:=\{\psi_i^b:i,b\}\cup{\cal T}\cup{\cal R}$, where ${\cal R}$ is again a set of pure states such that $\psi$ satisfies eq. (\ref{max_mixed_cond}). We define a prepare and measure scenario with $X=\sum_b d_b+|{\cal T}|+D-1$, $Y=(X^2-X)/2+1$, where measurements $y=1,...,Y-1$ are dichotomic and measurement $y=Y$ has $B$ outcomes, and consider the witness
\begin{equation}
W_{\bar{M}}(P)=W_\psi(P)-\sum_b\sum_{i=1}^{d_b}P(b|x=(i,b),y=Y).
\label{witness_POVMs}
\end{equation}
The maximum value of this witness is clearly $1-\frac{1}{D}$, achievable by preparing the states $\psi$, conducting the optimal dichotomic measurements to distinguish every pair of states in $\psi$ and also $M$ as the $Y^{th}$ POVM. 
Now, suppose that the maximum value of the witness is achieved by preparing states $\bar{\psi}$ and conducting POVM $\bar{M}$ for measurement $y=Y$. Since the witness $W_\psi(P)$ is saturated, there exists a(n) (anti-)unitary transformation $U$ such that $U\bar{\psi} U^\dagger =\psi$. Let us then consider the POVM $M':=U\bar{M}U^\dagger$. Then this POVM satisfies $\sum_b\tr(Z_b M_b')=0$, and thus, by the previous reasoning, $M'_b=M_b$, for all $b$. 
We demonstrate the above construction for self-testing on a specific extremal non-projective POVM in Section~\ref{app:examplePOVM} of the Appendix.  On the other hand, in Section \ref{app:robust_POVM} of the Appendix, we present a robust version of this argument. Namely, we show that, if $W_{\bar{M}}(P)>1-\frac{1}{D}-\epsilon$, then there exists a unitary or anti-unitary $U$ such that $U\tilde{M}_bU^\dagger=\bar{M}_a+\delta$, with $\delta=O(\epsilon^{1/16})$ in the $D=2$ case or $\delta=O(\epsilon^{1/32})$, otherwise.

The above construction allows one to self-test several extremal POVMs at a time: indeed, it suffices to add more terms of the form $\sum_b\sum_{i=1}^{d_{b|y}}P(b|x=(i,b,y),y)$ to (\ref{witness_POVMs}) and update the state preparation witness to self-test the states $\psi_{i}^{b|y}$ required to express the projector onto the kernel of the desired POVM $(M_{b|y}:b)$. In sum, we can devise a prepare-and-measure experiment to self-test as many pure states and extremal POVMs as we wish.

\vspace{10pt}

\noindent \emph{Dealing with higher dimensional leakages}

\vspace{10pt}

Throughout this Letter, we were assuming that the Hilbert space where the preparations took place had dimension $D$. In a realistic experiment, though, it is more plausible that Alice's preparations $\{\bar{\psi}_i\}_i$ actually live in $B(\C^{E})$, with $E>D$, possibly with $E=\infty$. Unfortunately, it is impossible to devise a $D$-dimensional experiment that self-tests preparations and measurements under the assumption that both objects act on a Hilbert space of dimension smaller than or equal to $E>D$, even if $E=D+1$ \footnote{Consider, e.g., a $1$-shot $D$-dimensional behavior $P(b|x,y)$ generated by the reference states $\{\ket{\psi_i}\}\in \mbox{span}(\ket{1},...,\ket{D})$ and measurements $\{(M_{b|y})_b:y\}\subset B(\C^D)$, and suppose that $\tr(\psi_1\psi_2)\not=0$. Then, one can obtain the same experimental behavior in a $D+1$-dimensional state if we replace $\psi_1$ by $\proj{D+1}$, and $M_{b|y}$ by $M_{b|y}+P(b|1,y)\proj{D+1}$. In this new representation, though, the first two preparations are orthogonal and thus not connected by a unitary or an anti-unitary transformation.}. Our results are, however, robust under the assumption that there exists a $D$-dimensional projector $\Pi_D$ such that the prepared states satisfy $\tr(\bar{\psi}_i\Pi_D)\geq 1-\delta,\forall i$. Indeed, as long as $\delta$ is small enough, it is easy to see that a close-to-maximal value of our linear witnesses implies the existence of an (anti-)isometry that approximately transforms the actual states and measurements into the reference ones \footnote{Note that the $D$-dimensional ensemble $\tilde{\psi}_i:=\frac{\Pi_D\psi_i\Pi_D}{\tr(\psi_i\Pi_D)}, i=1,...,N$ satisfies $\|\bar{\psi}_i-\tilde{\psi}_i\|_1\leq O(\sqrt{\delta}),\forall i$. Thus, the $D$-dimensional distribution $\tilde{P}(b|x,y):=\tr(\tilde{\psi}_i\tilde{M}_{b|y})$, with $\tilde{M}_{b|y}:=\Pi_DM_{b|y}\Pi_D$, is $O(\sqrt{\delta})$-away from the experimental $1$-shot distribution $P(b|x,y)$. It follows that $W(P)> W^\star -\epsilon$ implies that $W(\tilde{P})> W^\star -\epsilon-O(\sqrt{\delta})$. Thus, for $\epsilon,\delta$ small enough, there exists a(n) (anti-)unitary $U$ such that $U\tilde{\psi}U^\dagger\approx \psi_i$ $U\tilde{M}_{b|y}U^\dagger\approx M_{b|y}$. Consequently, the (anti-)isometry $V:=U\Pi_D$ approximately maps $\{\bar{\psi}_i\}_i$, $\{\bar{M}_{b|y}\}_{b,y}$ to the reference states and measurements $\{\psi_i\}_i$, $\{M_{b|y}\}_{b,y}$.}.

\vspace{10pt}

\noindent \emph{Conclusion}

\vspace{10pt}

In this Letter, we have completely characterized the limits of robust self-testing in the prepare-and-measure scenario, under a promise on the Hilbert space dimension of the prepared states. Namely, we have proven that, for any ensemble of pure quantum states and any set of extremal POVMs, one can devise a linear witness whose maximal value implies that the underlying states and measurements are related to the reference ones by a unitary or an anti-unitary transformation. 

Regrettably, the analytic robustness bounds we found, while exhibiting a reasonably good dependence on the Hilbert space dimension, scale with the experimental error $\epsilon$ as $\epsilon^{1/32}$. As such, they are impractical for realistic implementations. For small dimensions and a small number of preparation states, it might be feasible, though, to obtain more accurate predictions by combining the swap technique of \cite{swap} with standard semidefinite programming relaxations of the set $Q_D$ \cite{finite_hierarchy}, \cite{finite_hierarchy_2}.




\begin{acknowledgments}

\begin{wrapfigure}{r}[0cm]{2cm}
\begin{center}
\includegraphics[width=2cm]{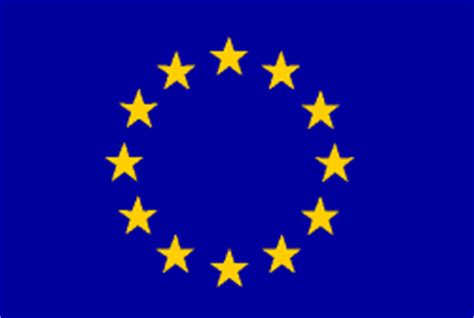}    
\end{center}
\end{wrapfigure} 


This project was funded within the QuantERA II Programme that has received funding from the European Union's Horizon 2020 research and innovation programme under Grant Agreement No 101017733, and from the Austrian Science Fund (FWF), project I-6004. T.V. acknowledges the support of the EU (QuantERA eDICT) and the National Research, Development and Innovation Office NKFIH (No.~2019-2.1.7-ERA-NET-2020-00003). M.A. acknowledges support by the Spanish Ministry of Science and Innovation (MCIN) with funding from the European Union Next Generation EU (PRTRC17.I1) and the Department of Education of Castilla y Le\'on (JCyL) through the QCAYLE project, as well as MCIN projects PID2020-113406GB-I00 and RED2022-134301-T.

\end{acknowledgments}

\bibliography{biblios_wigner}

\onecolumngrid
\begin{appendix}
\section{Proof of Lemma \ref{witness_norms}}
\label{app_witness_norms}
To prove Lemma \ref{witness_norms} (in the main text), it is convenient to define, given a set of quantum states $\rho:=\{\rho_i\}_{i=1}^N\subset B(\C^D)$, the vector $\bar{c}(\rho)=(c_{ij}(\rho))_{i>j}$, with
\begin{equation}
c_{ij}(\rho):=\sqrt{\alpha_i\alpha_j}\|\rho_i-\rho_j\|_1.
\label{c_def}
\end{equation}

To prove the theorem, we will make use of the following proposition
\begin{prop}
\label{prop_bound}
\begin{align}
\|c(\rho)\|_2^2\leq 2\left(1-\tr\left\{\left(\sum_i\alpha_i\rho_i\right)^2\right\}\right)\leq 2\left(1-\frac{1}{D}\right),
\label{bound_purity}
\end{align}
where the first inequality is tight iff 
\begin{equation}
\|\rho_i-\rho_j\|_1=2\sqrt{1-\tr\{\rho_i\rho_j\}},\forall i>j.
\label{scalar_trace}
\end{equation}
In particular, it is tight if $\{\rho_i\}_i$ are pure states. The second inequality is saturated iff 
\begin{equation}
\sum_i\alpha_i\rho_i=\frac{\id}{D}.
\label{id_res}
\end{equation}
\end{prop}

\begin{proof}
For any pair of states $\sigma,\omega$,
\begin{equation}
\|\sigma-\omega\|_1\leq 2\sqrt{1-F(\sigma,\omega)^2},
\end{equation}
where $F(\sigma,\omega)=\tr(|\sqrt{\sigma}\sqrt{\omega}|)$ denotes Uhlmann's fidelity \cite{Jozsa}. In turn, 
\begin{equation}
\tr(|\sqrt{\sigma}\sqrt{\omega}|)\geq \sqrt{\tr(\sigma\omega)},    
\end{equation}
\noindent and thus 
\begin{equation}
\|\sigma-\omega\|_1\leq 2\sqrt{1-\tr(\sigma\omega)}.
\label{bound_trace}
\end{equation}
Note that the upper bound is tight if $\sigma,\omega$ are pure states.

Substituting the bound (\ref{bound_trace}) in eq. (\ref{c_def}), we have that

\begin{align}
\|c(\rho)\|_2^2&\leq\sum_{i>j}\alpha_i\alpha_j4(1-\tr\{\rho_i\rho_j\})\nonumber\\
&\leq 2\sum_{i,j}\alpha_i\alpha_j(1-\tr\{\rho_i\rho_j\})\nonumber\\
&=2\left(1-\tr\left\{\left(\sum_i\alpha_i\rho_i \right)^2 \right\}\right)\nonumber\\
&\leq 2\left(1-\frac{1}{D}\right),
\label{norm2overlap}
\end{align}
where, in order to arrive at the last bound, we invoked the fact that $\sigma:=\sum_i\alpha_i\rho_i$ is a normalized $D$-dimensional quantum state, and, as such, its purity $\tr(\sigma^2)$ cannot be below $\frac{1}{D}$ \footnote{That this is indeed the minimum purity achievable in dimension $D$ can be verified by solving the problem $\min \sum_{i=1}^D\lambda_i^4$ subject to the constraint $\sum_{i=1}^D\lambda_i^2=1$, where $\{\lambda^2_i\}_i$ represent the eigenvalues of $\sigma$. This optimization problem can be solved, e.g., through Lagrange multipliers.}.
\end{proof}

We are now ready to prove Lemma \ref{witness_norms} in the main text. By assumption, $\psi$ are pure states satisfying $\sum_i\alpha_i\psi_i=\frac{\id}{D}$. By the above Proposition, it thus follows that
\begin{align}
\|\bar{c}(\psi)\|^2_2&=2\left(1-\tr\left\{\left(\sum_i\alpha_i\psi_i\right)^2\right\}\right)\nonumber\\
&=2\left(1-\frac{1}{D}\right).
\end{align}
Now, let $\rho=\{\rho_i\}_{i=1}^N\subset B(\C^D)$ be a set of quantum states; let $\{M^{(i,j)}:i>j, i,j=1,...,N\}$ be a collection of POVMs; and define $P(b|x,y):=\tr(\rho_xM^y_b)$. We have that

\begin{align}
W_\psi(P)=&\sum_{i>j}c_{i,j}(\psi)\sqrt{\alpha_i\alpha_j}\tr\{(\rho_i-\rho_j)M^{(i,j)}_1\}\nonumber\\
&\leq\sum_{i>j}c_{i,j}(\psi)\sqrt{\alpha_i\alpha_j}\|\rho_i-\rho_j\|_1\nonumber\\
&=\frac12\bar{c}(\psi)\cdot\bar{c}(\rho)\nonumber\\
&\leq \frac12 \|\bar{c}(\psi)\|_2\|\bar{c}(\rho)\|_2\nonumber\\
&\leq 1-\frac{1}{D}.
\label{intermezzo}
\end{align}
The first inequality comes from the well-known state discrimination inequality, namely that $\tr\{(\rho_i-\rho_j)M^{(i,j)}\} \le \frac12 \|\rho_i-\rho_j\|_1$ for any states $\rho_i,\rho_j$ and POVM elements $M^{(i,j)}$. The next inequality is just the Cauchy-Schwarz inequality, while the last one is a straightforward application of Proposition \ref{prop_bound}. This proves the first part of Lemma \ref{witness_norms} in the main text. With measurements $\{M^{(i,j)}:i>j\}$ that are optimal for state discrimination, the set of preparation states $\{\psi_i\}_i$ trivially saturates the bound.

Let us now assume that $W_\psi(P)=1-\frac{1}{D}-\epsilon$, for some $\epsilon \in [0,1-\frac1D]$. By eq. (\ref{intermezzo}), we have that 
\begin{equation}
\frac12\bar{c}(\psi)\cdot\bar{c}(\rho)\geq 1-\frac{1}{D}-\epsilon=\frac12\|\bar{c}(\psi)\|^2_2-\epsilon.
\end{equation}
This implies that
\begin{align}
\|\bar{c}(\psi)-\bar{c}(\rho)\|_2^2&=\|\bar{c}(\psi)\|_2^2+\|\bar{c}(\rho)\|_2^2-2\bar{c}(\psi)\cdot\bar{c}(\rho)\nonumber\\
&\leq 2\|\bar{c}(\psi)\|_2^2-2\bar{c}(\psi)\cdot\bar{c}(\rho))\leq 4\epsilon,
\label{vect_diff}
\end{align}
where, in the second line, we invoked the relation $\|\bar{c}(\psi)\|_2\geq \|\bar{c}(\rho)\|_2$, which follows from eq. (\ref{bound_purity}).



From eq. (\ref{intermezzo}) it also follows that
\begin{equation}
\|\bar{c}(\rho)\|_2\geq \|\bar{c}(\psi)\|_2-\frac{2\epsilon}{\|\bar{c}(\psi)\|_2}.
\end{equation}
Hence, for $2\epsilon\le\|\bar{c}(\psi)\|_2^2$, we have that
\begin{align}
\|\bar{c}(\rho)\|^2_2&\geq \left(\|\bar{c}(\psi)\|_2-\frac{2\epsilon}{\|\bar{c}(\psi)\|_2}\right)^2\nonumber\\
&=\|\bar{c}(\psi)\|_2^2+\frac{4\epsilon^2}{\|\bar{c}(\psi)\|_2^2}-4\epsilon\geq \|\bar{c}(\psi)\|_2^2-4\epsilon.
\label{bound_norm}
\end{align}
Eq. (\ref{norm2overlap}) states that
\begin{equation}
\|\bar{c}(\rho)\|^2_2 \leq \|\bar{d}(\rho)\|^2_2 \leq \|\bar{c}(\psi)\|_2^2,
\label{domin_norm}
\end{equation}
where the vector $\bar{d}(\rho)\equiv (d_{ij}(\rho))_{i>j}$ is defined through
\begin{equation}
\bar{d}_{ij}(\rho):=2\sqrt{\alpha_i\alpha_j}\sqrt{1- \tr\{\rho_i\rho_j\}}.    
\end{equation}
From the relation $\|\bar{c}(\psi)\|_2^2\geq \|\bar{d}(\rho)\|^2$  and eq. (\ref{bound_norm}) it thus follows that
\begin{equation}
\sum_{i>j}\alpha_i\alpha_j\left[4\left(1-\tr\{\rho_i\rho_j\}\right)-\|\rho_i-\rho_j\|^2_1\right]=\|\bar{d}(\rho)\|^2_2-\|\bar{c}(\rho)\|^2_2\leq 4\epsilon.
\label{part1}
\end{equation}
The left hand side of this equation is lower-bounded by 
\begin{equation}
\sum_{i>j}\alpha_i\alpha_j\left[2\sqrt{1-\tr\{\rho_i\rho_j\}}-\|\rho_i-\rho_j\|_1\right]^2=\|\bar{d}(\rho)-\bar{c}(\rho)\|_2^2.
\label{simplif}
\end{equation}
This follows from applying the relation $z^2-t^2\geq (z-t)^2$, valid for all $z\geq t\geq 0$, to $z=2\sqrt{1-\tr(\rho_i\rho_j)}$, $t=\|\rho_i-\rho_j\|_1$. 

Combining $\|\bar{d}(\rho)-\bar{c}(\rho)\|_2^2\leq 4\epsilon$ with eq. (\ref{vect_diff}), we have that
\begin{align}
\|\bar{d}(\rho)-\bar{d}(\psi)\|_2 &\leq \|\bar{d}(\rho)-\bar{c}(\rho)\|_2 + \|\bar{d}(\psi)-\bar{c}(\rho)\|_2 \\
&\leq 4\sqrt\epsilon,
\label{cross_prods_lemma}
\end{align}
where we made use of the triangle inequality and the identity $\bar{d}(\psi)=\bar{c}(\psi)$, valid by virtue of $\|\psi_i-\psi_j\|_1=2\sqrt{1-\tr\{\psi_i\psi_j\}}$. It follows that, if the experimental value of the witness is $\epsilon$-away from its maximal value, then for every $i\not=j$, 
\begin{equation}
\left|\sqrt{1-\tr\{\rho_i\rho_j\}}-\sqrt{1-\tr\{\psi_i\psi_j\}}\right|\leq \sqrt{\frac{2}{\alpha_i \alpha_j}} \sqrt\epsilon.
\end{equation}
Multiplying this expression by $\sqrt{1-\tr\{\rho_i\rho_j\}}+\sqrt{1-\tr\{\psi_i\psi_j\}}$, which is upperbounded by 2, we arrive at
\begin{equation}
\left|\tr\{\rho_i\rho_j\}-\tr\{\psi_i\psi_j\}\right|\leq \sqrt{\frac{8}{\alpha_i \alpha_j}} \sqrt\epsilon.
\end{equation}

It remains to prove that $\{\rho_i\}_i$ are approximately pure. To do so, we consider the inequality
\begin{equation}
\frac{1}{D}\leq \tr\left\{\left(\sum_i\alpha_i\rho_i\right)^2\right\},
\end{equation}
which is a consequence of Proposition \ref{prop_bound}. Expanding the right-hand side, we find that
\begin{align}
\frac{1}{D}&\leq \sum_i\alpha_i^2\tr\{\rho_i^2\}+2\sum_{i>j}\alpha_i\alpha_j\tr\{\rho_i\rho_j\}\nonumber\\
&=\sum_i\alpha_i^2\tr\{\rho_i^2\}-2\sum_{i>j}\alpha_i\alpha_j(1-\tr\{\rho_i\rho_j\})+\sum_{i,j}\alpha_i\alpha_j-\sum_i\alpha_i^2\nonumber\\
&= -\sum_i\alpha_i^2(1-\tr\{\rho_i^2\})-\frac{1}{2}\|\bar{d}(\rho)\|_2^2+1\nonumber\\
&\leq -\sum_i\alpha_i^2(1-\tr\{\rho_i^2\})-\frac{1}{2}\|\bar{c}(\rho)\|^2_2+1\nonumber\\
&\le -\sum_i\alpha_i^2(1-\tr\{\rho_i^2\})+2\epsilon-\frac{1}{2}\|\bar{c}(\psi)\|^2+1\nonumber\\
&=-\sum_i\alpha_i^2(1-\tr\{\rho_i^2\})+2\epsilon +\frac{1}{D}.
\label{casi_fin}
\end{align}
To arrive at the fourth line, we invoked eq. (\ref{domin_norm}); the fifth line follows from eq. (\ref{bound_norm}). Rearranging, we conclude that
\begin{equation}
\sum_i\alpha_i^2(1-\tr\{\rho_i^2\})\leq 2\epsilon.
\end{equation}
As $\epsilon$ tends to zero, the purity $\tr\{\rho_i^2\}$ of each preparation state $\rho_i$ is thus $\epsilon$-away from $1$, as the second line of eq. (\ref{basic_relation}) in the main text states.

Putting all together, we arrive at the following robustness result.
\begin{prop}
\label{prop_robust_overlaps}
Let $0 \le \epsilon \le 1-\frac{1}{D}$, and suppose that $W_\psi(P)=1-\frac{1}{D}-\epsilon$. Then, the prepared states $\{\rho_i\}_{i=1}^N$ satisfy:
\begin{gather}
\left|\tr\{\rho_i\rho_j\}-\tr\{\psi_i\psi_j\}\right|\leq \sqrt{\frac{8}{\alpha_i \alpha_j}} \sqrt\epsilon \quad\forall i\neq j \\
1-\tr\{\rho_i^2\}\leq \frac{2}{\sum_{i'}\alpha_{i'}^2}\epsilon \quad \forall i
\end{gather}
\end{prop}

\section{Wigner property for quantum state ensembles in dimension $D>2$}
\label{app:robust}
The previous section showed us how to robustly certify that the overlaps between the prepared states $(\rho_i)_i$ are similar to those of the reference pure states $(\psi_i)_i$. In this section we will show that reference states of the form $\Phi:=\{\psi_i\}_{i=1}^M\cup{\cal T}$, with ${\cal T}$ as defined below in eq. (\ref{def_T}) and $\tr(\psi_i Z_k)>0$, for $i=1,...,M$, $k=1,...,D$, satisfy the Wigner property robustly. That is, any ensemble of states $\{\rho_i\}_i$, pure or mixed, having similar overlaps is close in trace norm to $\psi$, modulo a unitary or anti-unitary transformation.

\subsection*{Notation to be used throughout the robustness proof}
The reference pure states we wish to self-test will be denoted as $\{\psi_k\}_k$; the auxiliary states constituting ${\cal T}$ are the following:
\begin{align}
&\ket{Z_k}:=\ket{k}, k=1,...,D,\nonumber\\
&\ket{X_k}:=\frac{1}{\sqrt{2}}(\ket{k}+\ket{k+1}), k=1,...,D-1,\nonumber\\
&\ket{Y_k}=\frac{1}{\sqrt{2}}(\ket{k}+i\ket{k+1}, , k=1,...,D-1,\nonumber\\
&\ket{XX_k}:=\frac{1}{\sqrt{3}}(\ket{k}+\ket{k+1}+\ket{k+2}),k=1,...,D-2,\nonumber\\
&\ket{YY_k}=\frac{1}{\sqrt{3}}(i\ket{k}+\ket{k+1}+i\ket{k+2}), k=1,...,D-2.
\label{def_T}
\end{align}
Given a symbol $A_k$, corresponding to an element of either $\{\psi_i\}_i$ or ${\cal T}$, we will use the symbol $A^{(0)}_k$ to refer to the actual physical state prepared in the experiment. E.g.: the physical preparation $Z_k^{(0)}$ corresponds to the reference state $Z_k$; we wish to self-test state $\psi_k$, but we actually prepared state $\psi_k^{(0)}$. In the course of the proof, we will use the super-index in brackets to keep track of the (anti-)unitary transformations that we will subject the prepared states to. Namely, $A_k^{(j)}=U_jA_k^{(j-1)}U^\dagger_j$, where $U_j$ is a unitary or anti-unitary operator. We will denote by $\Phi$ the union of the set of target states $\{\psi_i\}_{i=1}^M$ and ${\cal T}$. Correspondingly, $\Phi^{(0)}$ will denote the set of all the corresponding physical preparations, and $\Phi^{(j)}=U_j\Phi^{(j-1)}U_j^\dagger$.

\subsection*{Main result}
Our main result is the following lemma.
\begin{lemma}
\label{lemma_robust_state_self_testing}
Let $\Phi:=\{\psi_i\}_{i=1}^M\cup{\cal T}\subset \P$, with $\tr(\psi_j Z_k)>\frac{1}{f_j}$, for $j=1,...,M$, $k=1,...,D$, and let $\Phi^{(0)}=\{\psi^{(0)}_i\}_{i=1}^M\cup{\cal T}^{(0)}$ be quantum states (mixed or pure) such that
\begin{align}
&|\tr(\alpha^{(0)}\beta^{(0)})-\tr(\alpha\beta)|\leq \delta_o, \forall\alpha,\beta\in\Phi,\alpha\not=\beta.\label{real_overlaps}\\
&1-\tr((\alpha^{(0)})^2)\leq \delta_p, \forall \alpha\in\Phi,
\label{purity_bound}
\end{align}
Define the numbers $\delta_Z',\delta_Z,\delta_X',\delta_X,\delta_{XX},\delta_Y,\delta_{YY},\delta^\psi_{j}$ through the recurrent equations: 
\begin{align}
&\delta'_Z=2\delta_p+\delta_o,\nonumber\\
&\delta_Z=2\delta_p+2\sqrt{1-\left(1-\frac{\delta'_Z}{8}(D-1)^2\right)^2},\nonumber\\
&\delta'_X=\frac{\delta_Z}{2}+ \delta_p+\delta_o,\nonumber\\
&\delta_X=2\sqrt{2}\sqrt{\delta_X'}+2\delta_p,\nonumber\\
&\delta_Y=2\delta_p+2\sqrt{1-f\left(\left[\frac{\delta_X+\delta_Z}{2}+2\delta_p+2\delta_o\right]^2+\left[\frac{\delta_Z}{2}+\delta_p+\delta_o\right]^2,1-2\left[\frac{\delta_Z}{2}+\delta_p+\delta_o\right]\right)},\nonumber\\
&\delta_{XX}=2\sqrt{\frac{6}{5}\delta_X+\frac{9}{10}\delta_Z+\frac{21}{5}\delta_o},\nonumber\\
&\delta_{YY}=2\sqrt{\frac{3}{2}\delta_Z+2\delta_Y+7\delta_0},\nonumber\\
&\frac{(\delta^\psi_j)^2}{4}:=\left(D+\frac{(D-1)(1+\sqrt{2})f_j}{2\xi(D,f_j)}\right)\delta_Z+\frac{f_j(D-1)}{2\xi(D,f_j)}\delta_X+\frac{f_j(D-1)}{2\xi(D,f_j)}\delta_Y\nonumber\\
&+\left(D+\frac{(D-1)(1+\sqrt{2})f_j}{\xi(D,f_j)}+2\frac{f_j(D-1)}{\xi(D,f_j)}\right)\delta_o,
\end{align}
where the quantity $\xi(D,F)$ and the function $f$ are defined as
\begin{equation}
\xi(D,F):=\left(\frac{1}{1-\frac{D-1}{F}}\right)\left(1-\cos\left(\frac{\pi}{D}\right)\right),
\end{equation}
\begin{equation}
f(\mu,\nu):= \frac{1}{2}\left(\nu + \sqrt{\nu^2-4\mu}\right).
\label{def_f}
\end{equation}
Further suppose that the quantities above satisfy the relations:
\begin{align}
&\delta'_Z<\frac{1}{(D-1)^2},\nonumber\\
&\left[1-2\left(\frac{\delta_Z}{2}+\delta_p+\delta_o
\right)\right]^2-4\left(\left[\frac{\delta_X+\delta_Z}{2}+2\delta_p+2\delta_o\right]^2+\left[\frac{\delta_Z}{2}+\delta_p+\delta_o\right]^2\right)\geq 0,\nonumber\\
&-0.2938 + 0.7970\delta_Y+2.2555\delta_Z+0.5\delta_{XX}+7.1049\delta_o<0.
\end{align}
Then, there exists a unitary or anti-unitary transformation ${\cal U}$ such that
\begin{align}
&\|{\cal U}X^{(0)}_k{\cal U}^\dagger-X_k\|_1\leq \delta_X,\|{\cal U}Y^{(0)}_k{\cal U}^\dagger-Y_k\|_1\leq \delta_Y,k=1,...,D-1,\nonumber\\
&\|{\cal U}Z^{(0)}_k{\cal U}^\dagger-Z_k\|_1\leq \delta_Z,k=1,...,D,\nonumber\\
&\|{\cal U}XX^{(0)}_k{\cal U}^\dagger-XX_k\|_1\leq \delta_{XX},\|{\cal U}YY^{(0)}_k{\cal U}^\dagger-YY_k\|_1\leq \delta_{YY},k=1,...,D-2,\nonumber\\
&\|{\cal U}\psi^{(0)}_j{\cal U}^\dagger-\psi_j\|_1\leq \delta^\psi_j, j=1,...,M.
\end{align}
\end{lemma}

Note that, by Proposition \ref{prop_robust_overlaps}, conditions (\ref{real_overlaps}), (\ref{purity_bound}) hold if $W_\psi(P)$ is close to its maximum value. Indeed, if $W_\psi(P)=1-\frac{1}{D}-\epsilon$, by Proposition \ref{prop_robust_overlaps} we have that eqs. (\ref{purity_bound}) and (\ref{real_overlaps}) follow by taking 
\begin{gather}
\delta_p=\frac{2\epsilon}{\sum_i\alpha_i^2} \\
\delta_o=\frac{\sqrt{8\epsilon}}{\sqrt{\min_{i\not=j}\alpha_i\alpha_j}}.
\end{gather}

\subsection*{General tools}
The proof will use a couple of simple results extensively. One of them is the following:
\begin{remark}
\label{substitution_remark}
Let $\rho,\rho',\sigma,\sigma'$ be normalized quantum states satisfying the inequalities:
\begin{align}
&\|\rho-\rho'\|_1\leq 2\delta_1, \|\sigma-\sigma'\|_1\leq 2\delta_2,\nonumber\\
&\left|\tr(\rho\sigma)- \nu\right|\leq \delta_\nu,
\label{basic_subst}
\end{align}
for some $\nu\in \R$. Then, 
\begin{equation}
\left|\tr(\rho'\sigma')- \nu\right|\leq \delta_\nu+\delta_1+\delta_2.
\end{equation}
\end{remark}
\begin{proof}
First, note that any two normalized quantum states $\omega_1,\omega_2$, with $\|\omega_1-\omega_2\|_1\leq \epsilon$, satisfy $\|\omega_1-\omega_2\|_\infty\leq \frac{\epsilon}{2}$. Indeed, let $\sum_i\lambda^+_i\proj{\beta^+_i}-\sum_i\lambda^-_i\proj{\beta^-_i}$ be the spectral decomposition of $\omega_1-\omega_2$, with $\lambda_i^+,\lambda^-_j\geq 0$, for all $i,j$. Since $\tr(\omega_1-\omega_2)=0$, it follows that $\sum_i\lambda^+_i=\sum_i\lambda^-_i$. Now, 
\begin{align}
&\|\omega_1-\omega_2\|_\infty=\max\left(\max_i\lambda^+_i,\max_i\lambda_i^-\right)\leq \max\left(\sum_i\lambda_i^+,\sum_i\lambda_i^-\right)=\frac{1}{2}\sum_i\lambda_i^++\lambda_i^-=\frac{1}{2}\|\omega_1-\omega_2\|_1.
\end{align}
Eq. (\ref{basic_subst}) then follows by the obvious observation
\begin{equation}
\left|\tr(\rho'\sigma')- \nu\right|\leq\left|\tr(\rho\sigma)- \nu\right|+\|\rho-\rho'\|_\infty+\|\sigma-\sigma'\|_\infty.
\end{equation}
\end{proof}

The next remark follows from the relation between the trace distance and Uhlmann's fidelity, see the proof of Proposition \ref{prop_bound}.
\begin{remark}
\label{trace_dist_remark}
Let $\rho$ be normalized quantum state and $\ket{\phi}$, a normalized vector. Then, 
\begin{equation}
\|\rho-\proj{\phi}\|_1\leq 2\sqrt{1-\bra{\phi}\rho\ket{\phi}}.
\end{equation}
\end{remark}

Several times throughout the proof, we are faced with scenarios where we know the overlaps between an unknown quantum state $\rho$ and a number of characterized states. In these scenarios we might need to decide if said constraints are feasible, or use these constraints to lower bound the overlap of $\rho$ with another characterized state. In this regard, the following remark will be very useful.

\begin{remark}
\label{overlaps_remark}
Let $\rho\in B(\C^D)$, $\rho\geq 0$, and let $\{\beta_k\}_k\cup \{\beta\}\subset B(\C^D)$, with
\begin{equation}
\left|\tr(\rho \beta_k)- \lambda_k\right|\leq \delta_k,
\end{equation}
for $k=1,...,n$. Then, this set of constraints is infeasible if there exist $\{c_k\}_k\subset \R$ such that
\begin{equation}
C:=\sum_kc_k\beta_k\geq0,
\label{infeas_cert}
\end{equation}
and 
\begin{equation}
\sum_kc_k\lambda_k+|c_k|\delta_k<0.
\label{infeas_cond}
\end{equation}
This follows by noting that eq. (\ref{infeas_cond}) is an upper bound on $\tr(\rho C)$.

In addition, if there exist $\{d_k\}_k\subset \R$ such that
\begin{equation}
\beta-\sum_kd_k\beta_k\geq0,
\label{dual_feas}
\end{equation}
then 
\begin{equation}
\tr(\rho\beta)\geq \sum_kd_k\tr(\rho\beta_k)\geq \sum_kd_k\lambda_k -\sum_k|d_k|\delta_k.
\label{fid_lower}
\end{equation}
\end{remark}
Note that the task of finding $\{c_k\}_k\subset \R$ such that eq. (\ref{infeas_cert}), (\ref{infeas_cond}) hold can be cast as a semidefinite program (SDP) \cite{sdp}. Identifying $\{d_k\}_k\subset \R$ that satisfy (\ref{dual_feas}) and maximize the right-hand side of (\ref{fid_lower}) is also an SDP. In the proof, we will provide concrete examples of vectors $\vec{c},\vec{d}$ that do the job. This is so because we are aiming at a general proof of robustness. In particular cases where the values of $\delta_k$ are known, it is recommended to carry out an SDP to certify infeasibility or provide a better lower bound for the right-hand side of (\ref{fid_lower}).

\subsection*{Proof roadmap}
Given (\ref{real_overlaps}), (\ref{purity_bound}), the proof of robustness has the following structure:

\begin{enumerate}
    \item First, we prove that there exists a unitary $U$ that maps the states $(Z^{(0)}_k)_k$ to states that are $\delta_Z$-close (in trace norm) to $\{Z_k\}_k$. From now on, we thus consider the ensemble of states $\Phi^{(1)}:=U\Phi^{(0)}U^\dagger$.
    \item Using $Z^{(1)}_k\approx Z_k$, we prove that there exists a unitary $V$ so that $VZ_kV^\dagger=Z_k$, and $V X^{(1)}_k V^\dagger$ is $\delta_X$-close to $X_k$, for all $k$. Define, therefore, $\Phi^{(2)}:=V\Phi^{(1)} V^\dagger$.
    \item Using $Z_k^{(2)}\approx Z_k, X_k^{(2)}\approx X_k$, we show that $Y^{(2)}_k$ is $\delta_Y$-close to $Y''_k\in\{Y_k,Y_k^*\}$.
    \item Using $Z_k^{(2)}\approx Z_k, X_k^{(2)}\approx X_k$, we prove that $XX_k^{(2)}$ is $\delta_{XX}$-close to $XX_k$. 
    \item From the overlaps $\tr(Z_kYY_k)$, $\tr(Z_{k+1}YY_k)$, $\tr(Y_kYY_k)$, $\tr(Y_{k+1}YY_k)$, $\tr(XX_k YY_k)$, we show that, provided that $\delta_Y,\delta_Z,\delta_{XX}$ are small enough, either $Y''_k=Y_k$, for all $k$, or $Y''_k=\overline{Y_k}$, for all $k$. In the first case, let $W$ be the identity transformation. In the second case, let $W$ be the anti-unitary transformation that leaves the computational basis invariant. Define $\Phi^{(3)}=W\Phi^{(2)}W^\dagger$. Then we have that $Z^{(3)}_k \approx Z_k$, $X^{(3)}_k\approx X_k$, $Y^{(3)}_k\approx Y_k$, for all $k$. We also prove that $YY_k^{(3)}$ is $\delta_{YY}$-close to $YY_k$, for all $k$.
    \item As noted in \cite{Geher2014}, the set of operators $\{X_k,Y_k,Z_k\}$ is tomographically complete for all pure states $\psi$ with $\tr(\psi Z_k)>0$, for $k=1,...,D$. We make this observation robust, thus proving that, for $\delta_X,\delta_Y,\delta_Z,\delta_o$ small enough, $\psi^{(3)}_k$ is $\delta_k^\psi$-close to $\psi_k$, for $k=1,...,M$.
\end{enumerate}
If at some point during the proof the reader feels lost, we recommend him/her to come back to the roadmap to get a notion of where the proof is going to.

\subsection{Recovering $\{Z_k\}_k$}
In this subsection we show that, if eqs. (\ref{real_overlaps}), (\ref{purity_bound}) hold, then the state ensemble $(Z_k^{(0)})_k$ can be well approximated (modulo unitaries) by $(Z_k)_k$.

First, we note that a state $\rho$ of high purity $\tr(\rho^2)\approx 1$ can be well approximated by a pure state.
\begin{lemma}
Let $\rho\in B(\C^D)$ be a quantum state. If $\tr(\rho^2)\geq 1-\delta$, then there exists a normalized state $\ket{\phi}\in \C^D$ such that $\|\rho-\proj{\phi}\|_1<2\delta$.
\label{purity_lemma}
\end{lemma}

\begin{proof}
Let $\rho=\sum_i p(i)\proj{\phi_i}$ be the eigen-decomposition of $\rho$, with $p(1)\geq p(2)\geq...\geq p(D)$. The condition $\tr(\rho^2)\geq 1-\delta$ implies that $\langle p(i)\rangle_p\geq 1-\delta$, and so $p(1)\geq 1-\delta$. Taking $\ket{\phi}=\ket{\phi_1}$, we thus find that $\|\rho-\proj{\phi}\|_1=2(1-p(1))\leq 2\delta$.
\end{proof}

Applying the lemma with $\delta=\delta_p$ to the states $\{Z_k^{(0)}\}_k$, we find that we can replace them by the pure states $\{Z'_k\}_k$, at a cost $2\delta_p$. By Remark \ref{substitution_remark}, the overlap between the states $\{Z'_k\}_k$ satisfies
\begin{equation}
\left|\tr(Z'_jZ'_k)-\delta_{jk}\right|\leq 2\delta_p+\delta_o=:\delta'_Z.
\label{def:delta_z_prime}
\end{equation}

Next, we need to prove that the pure states $\{Z'_k\}_{k=1}^D$ can be replaced by an orthonormal basis without introducing too much noise. 

\begin{lemma}
Let $\{\phi_i\}_{i=1}^D\subset B(\C^D)$ be a set of pure quantum states, with
\begin{equation}
\tr(\phi_i\phi_j)\leq \delta, i\not=j.
\label{over_lemma}
\end{equation}
If $\delta<\frac{1}{(D-1)^2}$, then there exists an orthonormal basis $\{\phi'_i\}_{i=1}^D$, with
\begin{equation}
\|\phi_i-\phi_i'\|_1\leq 2\sqrt{1-\left(1-\frac{\delta}{8}(D-1)^2\right)^2},i=1,...,D.
\end{equation}
\label{lemma_basis}
\end{lemma}

\begin{proof}
The idea of the proof is to consider the state discrimination problem with states $\{\phi_i\}_{i=1}^D$. That is: we wish to identify the POVM $\{M_i\}_i$ such that the probability of a successful discrimination $\sum_i\tr(M_i\phi_i)$ is maximized. If $\{\phi_i\}_{i=1}^D$ are linearly independent, then the optimal POVM is known to be an orthonormal measurement. Moreover, a good approximation to the optimal measurement is given by Hausladen and Wootter's `pretty good measurement'\cite{square_root}. We will identify the rank-$1$ projections of this measurement with the new states $\{\phi'_i\}_i$.

The first step is thus to verify that $\{\phi_i\}_{i=1}^D$ are linearly independent. Let $G$ be the Gram matrix $G$ of this set of vectors, i.e., $G_{ij}=\braket{\phi_i}{\phi_j}$. By relation (\ref{over_lemma}), we can write it as $G=\id+\sqrt{\delta}R$, where $R$ is a Hermitian matrix with zero diagonal and non-diagonal elements of modulus at most $1$. Now, the norm of $R$ is at most $D-1$. Indeed, let $\ket{c}$ be an arbitrary normalized vector. Then,
\begin{equation}
\bra{c}R\ket{c}\leq\sum_{i\not=j}|c_i||c_j|\leq\|\tilde{R}\|=D-1,
\end{equation}
where $\tilde{R}=\sum_{i,j}\ket{i}\bra{j}-\id$. Thus $G$ is positive definite (and hence $\{\phi_i\}_i$ are linearly independent) if $1-\sqrt{\delta}(D-1)>0$.

Let $\Gamma:=\sum_i \ket{\phi_i}\bra{i}$, and note that $\Gamma^\dagger \Gamma=G$. The square root measurement requires using measurement vectors $\{\phi'_i\}_i$ such that
\begin{equation}
M:=\sum_i\ket{\phi_i'}\bra{i}=\Gamma G^{-\frac{1}{2}}.
\end{equation}
As promised, $M^\dagger M=\id$, and so the measurement is indeed a von Neumann measurement. Moreover, $\tr(\phi'_i\phi_i)=((\Gamma^\dagger M)_{ii})^2=(G^{\frac{1}{2}})^2_{ii}$. 

Now, 
\begin{equation}
\left\|G^{1/2}-\id-\sqrt{\delta}\frac{R}{2}\right\|\leq \frac{\delta}{8}\|R\|^2\leq \frac{\delta}{8}(D-1)^2.
\end{equation}
Taking into account that $R$ has zero diagonal elements, we have that $\tr(\phi'_i\phi_i)\geq \left(1-\frac{\delta}{8}(D-1)^2\right)^2$. From this relation and Remark \ref{trace_dist_remark}, the statement of the proposition follows.
\end{proof}

Thus we apply Lemma \ref{lemma_basis} to the set of pure states $\{Z'_k\}_{k=1}^D$ with $\delta=\delta'_Z$. We obtain the orthonormal basis $\{Z_k''\}$. Let $U$ be the unitary such that $U\ket{Z'_k}=\ket{Z_k}$, and define $\Phi^{(1)}_k:=U \Phi^{(0)}_k U^\dagger$. We thus have that
\begin{equation}
\|Z_k^{(1)}-Z_k\|_1\leq 2\delta_p+2\sqrt{1-\left(1-\frac{\delta'_Z}{8}(D-1)^2\right)^2}=:\delta_Z,
\label{def:delta_z}
\end{equation}
under the condition that 
\begin{equation}
\delta'_Z<\frac{1}{(D-1)^2}.
\label{cond:1}
\end{equation}

\subsection{Recovering $\{X_k\}_k$}
Here we will show how to replace, up to a small cost, the states $\{X^{(1)}_k\}_k$ by the exact reference states $\{X_k\}_k$. 

Like in the previous section, our first step is to convert $\{X^{(1)}_k\}$ into the pure states $\{X'_k\}_k$ via lemma (\ref{purity_lemma}), at a cost $2\delta_p$ in trace norm. By Remark \ref{substitution_remark}, for any $k$, the overlap between $X'_k$ and $Z_k, Z_{k+1}$ satisfies
\begin{equation}
|\tr(X'_kZ_j)-\frac{1}{2}|\leq \frac{\delta_Z}{2}+ \delta_p+\delta_o=:\delta'_X.
\label{def:delta_X_prime}
\end{equation}

Consider a unitary $V$ such that $V Z_kV^\dagger=\proj{Z_k}$, $VX_k'V^\dagger=\proj{X_k''}$, with $\braket{Z_k}{X''_k}$, $\braket{Z_{k+1}}{X_k''}$ real and non-negative \footnote{Such a unitary $V$ can be constructed as follows: let $X_k'=\proj{X'_k}$, and note that we can choose the overall phase of $\ket{X'_k}$ so that $\braket{Z_k}{X_k'}\in \R^+\cup\{0\}$, $\braket{Z_{k+1}}{X_k'}=r_ke^{i\theta_k}$, $r_k\in\R^+\cup \{0\}$. Then, the unitary $V\ket{Z_k}:=e^{-i\sum_{j<k}\theta_j}\ket{Z_k}$ maps $\ket{X'_k}$ to a state $\ket{\tilde{X}_k}$ where the complex numbers $\braket{Z_k}{\tilde{X}_k}$, $\braket{Z_{k+1}}{\tilde{X}_k}$ have the same phase $\tilde{\theta}_k$. Thus, the states $\ket{X_k''}:=e^{-i\tilde{\theta}_k}V\ket{\tilde{X}_k}$ satisfy $\braket{Z_k}{X''_k},\braket{Z_{k+1}}{X_k''}>0$.}, and define $A^{(2)}:= VA^{(1)}V^\dagger$. We hence have that
\begin{equation}
\braket{Z_{k}}{X''_k}=\sqrt{\frac{1}{2}+\nu_1}, \braket{Z_{k+1}}{X''_k}=\sqrt{\frac{1}{2}+\nu_2},    
\end{equation}
with $|\nu_j|\leq \delta_X'$, for $j=1,2$. Thus, provided that $\delta_X'\leq \frac{1}{2}$, we arrive at
\begin{equation}
\tr(X_k X''_k)=\frac{1}{2}\left(\sqrt{\frac{1}{2}+\nu_1}+\sqrt{\frac{1}{2}+\nu_2}\right)^2\geq 1-2\delta_X'.
\end{equation}
Note that, even in the case $\delta_X'> \frac{1}{2}$, the above equation holds. It follows, by Remark \ref{trace_dist_remark}, that $\|X''_k-X_k\|_1\leq 2\sqrt{2}\sqrt{\delta'_X}$.
Putting all together, we have that
\begin{equation}
\|X^{(2)}_k-X_k\|_1\leq 2\sqrt{2}\sqrt{\delta_X'}+2\delta_p=:\delta_X.
\label{def:delta_X}
\end{equation}
\subsection{Recovering $\{Y_k\}_k$ up to conjugation}
In this section we will prove that, for each $k$, $Y_k^{(2)}$ can be replaced by either $Y_k$ or $\overline{Y_k}$.

\begin{lemma}
\label{lemma_Ys}
Let $\ket{\phi}$ satisfy
\begin{align}
&\left||\braket{\phi}{Z_j}|^2-\frac{1}{2}\right|\leq \bar{\delta}_Z,j=k, k+1,\nonumber\\
&\left||\braket{\phi}{X_k}|^2-\frac{1}{2}\right|\leq\bar{\delta}_X,
\end{align}
with
\begin{equation}
(1-2\bar{\delta}_Z)^2-4\{\bar{\delta}_Z^2+(\bar{\delta}_X+\bar{\delta}_Z)^2\}\geq 0.
\end{equation}
Then, 
\begin{equation}
\min\left(\|\proj{\phi}-\proj{Y_k}\|_1, \|\proj{\phi}-\proj{\overline{Y_k}}\|_1\right)\leq 2\sqrt{1-f\left(\bar{\delta}_Z^2+(\bar{\delta}_X+\bar{\delta}_Z)^2, 1-2\bar{\delta}_Z\right)},
\end{equation}
with
\begin{equation}
f(\mu,\nu):= \frac{1}{2}\left(\nu + \sqrt{\nu^2-4\mu}\right).
\end{equation}
\end{lemma}

\begin{proof}
From the assumptions of the lemma, we have that
\begin{align}
&\left|\tr\left(\frac{Z_k-Z_{k+1}}{2}\phi\right)\right|\leq \bar{\delta}_Z,\nonumber\\
&\left|\tr\left(\frac{2X_k-Z_k-Z_{k+1}}{2}\phi\right)\right|\leq \bar{\delta}_Z+\bar{\delta}_X.
\end{align}
Since 
\begin{equation}
\ket{Y_k}\bra{\overline{Y_k}}=\frac{Z_k-Z_{k+1}}{2}+i\left(\frac{2X_k-Z_k-Z_{k+1}}{2}\right),
\end{equation}
it follows that
\begin{equation}
|\braket{\phi}{Y_k}||\braket{\phi}{\overline{Y_k}}|=\left|\tr(\ket{Y_k}\bra{\overline{Y_k}}{\phi})\right|\leq\sqrt{\bar{\delta}_Z^2+(\bar{\delta}_X+\bar{\delta}_Z)^2}=:\bar{\delta}_Y.
\end{equation}

On the other hand, note that
\begin{equation}
1\geq|\braket{\phi}{Y_k}|^2+|\braket{\phi}{\overline{Y_k}}|^2=|\braket{\phi}{Z_k}|^2+|\braket{\phi}{Z_{k+1}}|^2\geq 1-2\bar{\delta}_Z.
\end{equation}
Defining $s:= |\braket{\phi}{Y_k}|^2$, $t:= |\braket{\phi}{\overline{Y_k}}|^2$, we arrive at the system of equations
\begin{equation}
st=\mu, s+t=\nu,
\end{equation}
with $0\leq\mu\leq \bar{\delta}_Y^2$, $1-2\bar{\delta}_Z\leq\nu\leq 1$. Solving it, we find that 
\begin{equation}
\max(s,t)=\frac{1}{2}\left(\nu + \sqrt{\nu^2-4 \mu}\right).    
\end{equation}
Thus, provided that $(1-2\bar{\delta}_Z)^2-4\bar{\delta}_Y^2\geq 0$, $\max(s,t)=f(\bar{\delta}_Y, 1-2\bar{\delta}_Z)$, as defined in (\ref{def_f}). It follows that the trace distance of one of the projectors $Y_k,\overline{Y_k}$ to $\phi$ is at most twice the square root of one minus the right-hand-side of the equation.
\end{proof}

We apply Lemma \ref{lemma_Ys} to the states $\{Y'_k\}$ that result from purifying $\{Y^{(2)}_k\}_k$; hence, $\bar{\delta}_Z=\frac{\delta_Z}{2}+\delta_p+\delta_o$, $\bar{\delta}_X=\frac{\delta_X}{2}+\delta_p+\delta_o$. The lemma implies that there exist states $\{Y''_k\}_k$, with $Y''_k\in\{Y_k,\overline{Y_k}\}$, such that
\begin{equation}
\|Y^{(2)}_k-Y''_k\|_1\leq 2\delta_p+2\sqrt{1-f\left(\left[\frac{\delta_X+\delta_Z}{2}+2\delta_p+2\delta_o\right]^2+\left[\frac{\delta_Z}{2}+\delta_p+\delta_o\right]^2,1-2\left[\frac{\delta_Z}{2}+\delta_p+\delta_o\right]\right)}=:\delta_Y,
\label{def:delta_Y}
\end{equation}
provided that
\begin{equation}
\left[1-2\left(\frac{\delta_Z}{2}+\delta_p+\delta_o
\right)\right]^2-4\left(\left[\frac{\delta_X+\delta_Z}{2}+2\delta_p+2\delta_o\right]^2+\left[\frac{\delta_Z}{2}+\delta_p+\delta_o\right]^2\right)\geq 0.
\label{cond:3}
\end{equation}

\subsection{Recovering $\{XX_k\}_k$}
In this section, we wish to prove that $XX_k^{(2)}$ is approximated by $XX_k$, for all $k$. To achieve this, we propose the following lemma:

\begin{lemma}
Let $\phi$ be a state such that $\tr(\phi X_k), \tr(\phi X_{k+1})$ equal $\frac{2}{3}$ up to precision $\bar{\delta}_X$ and $\tr(\phi Z_{k+j})$ equals $\frac{1}{3}$ up to precision $\bar{\delta}_Z$, for $j=0,1,2$. Then,
\begin{equation}
\tr(\phi XX_k)\geq 1-\frac{12}{5}\bar{\delta}_X-\frac{9}{5}\bar{\delta}_Z.
\end{equation}
\end{lemma}
\begin{proof}
By Remark \ref{overlaps_remark}, it suffices to note that
\begin{equation}
\proj{XX_k}-\frac{6}{5}(X_k+X_{k+1})+\frac{3}{5}(Z_k+Z_{k+1}+Z_{k+2})\geq0.  
\end{equation}
\end{proof}

Applying this lemma to the state $XX_k^{(2)}$, with $\bar{\delta}_X=\frac{\delta_X}{2}+\delta_o$, $\bar{\delta}_Z=\frac{\delta_Z}{2}+\delta_o$, we have that
\begin{equation}
\tr(XX_k^{(0)} XX_k)\geq 1-\frac{6}{5}\delta_X-\frac{9}{10}\delta_Z-\frac{21}{5}\delta_o.
\end{equation}
Thus, by Remark \ref{trace_dist_remark},
\begin{equation}
\|XX_k^{(0)}-XX_k\|_1\leq 2\sqrt{\frac{6}{5}\delta_X+\frac{9}{10}\delta_Z+\frac{21}{5}\delta_o}=:\delta_{XX}.
\label{def:delta_XX}
\end{equation}

\subsection{Recovering $YY_k$}
We next prove that, for any $k$, either $Y''_k=Y_k$, $Y''_{k+1}=Y_{k+1}$ or $Y''_k=\overline{Y_k}$, $Y''_{k+1}=\overline{Y_{k+1}}$. By induction it will follow that, either $Y''_k=Y_k$ for all $k$, or $Y''_k=\overline{Y_k}$, for all $k$. We require the following lemma:

\begin{lemma}
\label{lemma_Y}
Let $\phi$ be a state such that $\tr(\phi Y_k)$, $\tr(\phi \overline{Y_{k+1}})$ respectively equal $\frac{2}{3}$ and $0$ up to precision $\bar{\delta}_Y$; $\tr(\phi Z_{k+j})$ equal $\frac{1}{3}$ up to precision $\bar{\delta}_Z$, for $j=0,1,2$, and $\tr(\phi XX_k)$ equals $\frac{5}{9}$ up to precision $\bar{\delta}_{XX}$. Then it must be the case that
\begin{equation}
-0.2938 + 1.5939\bar{\delta}_Y+4.5110\bar{\delta}_Z+\bar{\delta}_{XX}\geq0.
\end{equation}
\end{lemma}
Conjugating all the relations involved, one concludes that the lemma also holds true if we replace the overlaps $\tr(\phi Y_k)$, $\tr(\phi \overline{Y_{k+1}})$ by $\tr(\phi \overline{Y_k})$, $\tr(\phi Y_{k+1})$ in the statement of the lemma. 
\begin{proof}
By Remark \ref{overlaps_remark}, it is enough to realize that 
\begin{equation}
\sum_{j=1}^2\lambda_jZ_{k+j}+\mu\proj{Y_k}+\nu \proj{\overline{Y_{k+1}}}-\proj{XX_k}\geq0,
\label{YYs}
\end{equation}
with $\lambda=(1.8149,1.9415,0.2167)$, $\mu=-1.5939$, $\nu=2.9170$.
\end{proof}

Note that $\tr(YY_k\cdot Y_k)=\frac{2}{3}$, $\tr(YY_k\cdot Y_{k+1})=0$, $\tr(YY_k\cdot Z_{j+k})=\frac{1}{3}$, for $j=0,1,2$, and $\tr(YY_k\cdot XX_k)=\frac{5}{9}$. If we wish to refute that $Y''_k=Y_k$, $Y''_{k+1}=\overline{Y_{k+1}}$ (or, by the remark at the end of the lemma, $Y''_k=\overline{Y_k}$, $Y''_{k+1}=Y_{k+1}$), we can thus apply Lemma \ref{lemma_Y} to the state $\phi=YY_k$, with $\bar{\delta}_Y=\frac{\delta_Y}{2}+\delta_o$, $\bar{\delta}_Z=\frac{\delta_Z}{2}+\delta_o$, $\bar{\delta}_{XX}=\frac{\delta_{XX}}{2}+\delta_o$. 

We find that there is a contradiction whenever
\begin{equation}
-0.2938 + 0.7970\delta_Y+2.2555\delta_Z+0.5\delta_{XX}+7.1049\delta_o<0.
\label{cond:5}
\end{equation}
In the relation above holds, then it cannot be that $Y''_k=Y_k$, $Y''_k=\overline{Y_k}$, or $Y''_k=\overline{Y_k}$, $Y''_k=Y_k$. There are thus two possibilities: either $Y_k''=Y_k$, for $k=1,...,D-1$, or $Y_k''=\overline{Y_k}$, for all $k=1,...,D-1$. In the first case, let $W$ be the identity matrix in dimension $D$; in the second case, let $W$ be the anti-unitary that leaves invariant the vectors in the computational basis, i.e., $W\ket{Z_k}=\ket{Z_k}$. Define $\Phi^{(3)}=W\Phi^{(2)}W^\dagger$. From the above it follows that, as long as conditions (\ref{cond:1}), (\ref{cond:3}), (\ref{cond:5}) are satisfied, 
\begin{equation}
\|X_k-X_k^{(3)}\|_1\leq \delta_X, \|Y_k-Y_k^{(3)}\|_1\leq \delta_Y,\|Z_k-Z_k^{(3)}\|_1\leq \delta_Z,\|XX_k-XX_k^{(3)}\|_1\leq \delta_{XX}.
\end{equation}

It remains to bound $\|YY^{(3)}_k- YY_k\|_1$. This is accounted for by the following lemma:

\begin{lemma}
Let $\phi$ be a state such that $\tr(\phi Y_k)$, $\tr(\phi Y_{k+1})$ respectively equal $\frac{2}{3}$ and $0$ up to precision $\bar{\delta}_Y$ and $\tr(\phi Z_{k+j})$ equals $\frac{1}{3}$ up to precision $\bar{\delta}_Z$, for $j=0,1,2$. Then,
\begin{equation}
\tr(\phi YY_k)\geq 1-3\bar{\delta}_Z-4\bar{\delta}_Y.
\end{equation}
\end{lemma}
\begin{proof}
By Remark \ref{overlaps_remark}, it suffices to note that
\begin{equation}
\proj{YY_k}-Z_k+Z_{k+1}+Z_{k+2}+2(Y_k-Y_{k+1})\geq0.  
\end{equation}
\end{proof}
Applying the lemma with $\phi=YY^{(3)}_k$, $\bar{\delta}_X=\frac{\delta_X}{2}+\delta_0$, $\bar{\delta}_Z=\frac{\delta_Z}{2}+\delta_0$, we find that
\begin{equation}
\tr(YY_k^{(3)}\cdot YY_k)\geq 1 - \frac{3}{2}\delta_Z-2\delta_Y-7\delta_0.
\end{equation}
Thus,
\begin{equation}
\|YY^{(3)}_k- YY_k\|_1\leq 2\sqrt{\frac{3}{2}\delta_Z+2\delta_Y+7\delta_0}=:\delta_{YY}.
\label{def:delta_YY}
\end{equation}

\subsection{Recovering $\{\psi_i\}_{i=1}^M$}
By Remark \ref{substitution_remark}, we have that, for $j=1,...,M$,
\begin{align}
&|\tr(\psi_j^{(3)}X_k)- \tr(\psi_jX_k)|\leq \frac{\delta_X}{2}+\delta_o,\nonumber\\
&|\tr(\psi_j^{(3)}Y_k)- \tr(\psi_jY_k)|\leq \frac{\delta_Y}{2}+\delta_o,\nonumber\\
&|\tr(\psi_j^{(3)}Z_k)- \tr(\psi_jZ_k)|\leq \frac{\delta_Z}{2}+\delta_o.
\end{align}
In this section of the proof, we will prove that the above implies that $\psi_j^{(3)}$ is close to $\psi_j$ in trace norm.

To do so, we will argue that $\tr(\psi_j\psi^{(3)}_j)\approx 1$. By Remark \ref{overlaps_remark}, it suffices to find $\{\lambda_k\}_k,\{\mu_k\}_k,\{\nu_k\}_k$ such that the operator
\begin{equation}
H=\sum_{k=1}^D\lambda_i\proj{Z_k}+\sum_{i=1}^{D-1}\mu_i\proj{X_k}+\nu_i\proj{Y_k}
\label{hamiltonian}
\end{equation}
satisfies $\proj{\psi_j}-H\geq 0$, $\tr(\psi_j H)=1$. This will imply that $\tr(\psi_j\psi_j^{(3)})$ is lower-bounded by
\begin{equation}
1-\left(\frac{\delta_Z}{2}+\delta_o\right)\sum_{k=1}^D|\lambda_i|-\left(\frac{\delta_X}{2}+\delta_o\right)\sum_{i=1}^{D-1}|\mu_i|-\left(\frac{\delta_Y}{2}+\delta_o\right)\sum_{i=1}^{D-1}|\nu_i|.
\label{lower_bound_fid_psi}
\end{equation}

The next lemma is exactly what we need:

\begin{lemma}
\label{lemma_hamiltonian_tomography}
Let $\ket{\psi}$ be a quantum state such that $\tr(\psi Z_k)\geq \frac{1}{f}>0$, for all $k\in\{1,...,D\}$. Then, there exist coefficients $\{\lambda_i\}_i, \{\mu_i\}_i, \{\nu_i\}_i$ such that the operator
\begin{equation}
H=\sum_{k=1}^D\lambda_i\proj{Z_k}+\sum_{i=1}^{D-1}\mu_i\proj{X_k}+\nu_i\proj{Y_k}
\label{tomo_hamil}
\end{equation}
satisfies
\begin{equation}
\proj{\psi}-H\geq 0,\tr(H\proj{\psi})=1.
\end{equation}
Moroever, 
\begin{align}
&\sum_k|\lambda_k|\leq D+\frac{(D-1)(1+\sqrt{2})f}{\xi(D,f)}\nonumber\\
&\sum_k|\mu_k|\leq\frac{f(D-1)}{\xi(D,f)},\nonumber\\
&\sum_k|\nu_k|\leq \frac{f(D-1)}{\xi(D,f)}.
\label{error_control}
\end{align}
with
\begin{equation}
\xi(D,f):=\left(\frac{1}{1-\frac{D-1}{f}}\right)\left(1-\cos\left(\frac{\pi}{D}\right)\right)= O\left(\frac{1}{D^2}\right).
\label{def_xi}
\end{equation}
\end{lemma}
To prove the lemma, we need the following technical result.
\begin{prop}
\label{prop:usual_techniques}
The spectrum of the operator
\begin{equation}
\tilde{H}=\frac{1}{2}\sum_{k=1}^{D-1}\ketbra{Z_k}{Z_{k}}+\ketbra{Z_{k+1}}{Z_{k+1}}-\ketbra{Z_{k+1}}{Z_k}-\ketbra{Z_k}{Z_{k+1}}
\label{h_tilde}
\end{equation}
is $\{1-\cos\left(\frac{\pi j}{D}\right):j=0,...,D-1\}$. The eigenvector corresponding to the eigenvalue $0$ is $\ket{\tilde{\psi}_0}:=\frac{1}{\sqrt{D}}\sum_{k=1}^{D}\ket{Z_k}$.
\end{prop}
The proof of this proposition is presented in Appendix \ref{proof_usual_techniques}.

\emph{Proof of Lemma \ref{lemma_hamiltonian_tomography}.}  
Let $\ket{\psi}=\sum_kc_k\ket{Z_k}$, and define $C:=\sum_k c_k\proj{Z_k}$. Then it holds that $\ket{\psi}=\sqrt{D} C\ket{\tilde{\psi}_0}$, where $\ket{\tilde{\psi}_0}$ is the eigenvector of $\tilde{H}$, as defined in Proposition \ref{prop:usual_techniques}, with zero eigenvalue. Next, consider the operator $H':=(C^{-1})^\dagger\tilde{H}C^{-1}$. $H'$ is obviously positive semidefinite. Due to the invertibility of $C^{-1}$, the operator $H'$ has exactly one eigenvector with eigenvalue $0$, namely $\ket{\psi}$. Indeed, 
\begin{equation}
H'\ket{\psi}=(C^{-1})^\dagger\tilde{H}C^{-1}(\sqrt{D}C\ket{\tilde{\psi_0}})=\sqrt{D}(C^{-1})^\dagger\tilde{H}\ket{\tilde{\psi_0}}=0.
\end{equation}
Next we need to lower bound the gap of $H'$, namely, the difference between its second lowest and lowest eigenvalues. Define $\tilde{\lambda}_k:=1-\cos(\pi k/D)$, for $k=0,...,D-1$. According to Proposition \ref{prop:usual_techniques}, those are the eigenvalues of $\tilde{H}$ in increasing order, with eigenvectors $\{\ket{\tilde{\psi}_k}\}_k$. We have that
\begin{equation}
H'=\sum_{k=1}^{D-1}\tilde{\lambda}_k(C^{-1})^\dagger\proj{\tilde{\psi}_k}C^{-1}.
\end{equation}
Call $\lambda_1$ the smallest non-zero eigenvalue of $H'$. From the above equation, it follows that
\begin{equation}
\lambda_1=\min_\phi\bra{\phi}\left(\sum_k\tilde{\lambda}_k(C^{-1})^\dagger\proj{\tilde{\psi}_k}C^{-1}\right)\ket{\phi},
\label{gap}
\end{equation}
where the maximization takes place over normalized vectors $\ket{\phi}\in \mbox{span}\{(C^{-1})^\dagger\ket{\tilde{\psi}_k}:k=1,...,D-1\}$. 

Let therefore $\ket{\phi}=\sum_k d_k(C^{-1})^\dagger\ket{\tilde{\psi}_k}$. Then we have that
\begin{equation}
\bra{\phi}H'\ket{\phi}\geq \tilde{\lambda}_1\sum_{k=1}^{D-1}\bra{\phi}(C^{-1})^\dagger\proj{\tilde{\psi}_k}C^{-1}\ket{\phi}=\tilde{\lambda}_1\sum_{i,j,k=1}^{D-1}d_i^*G_{ik}G_{kj}d_j=\tilde{\lambda}_1\bra{d}G^2\ket{d},
\label{approx_gap}
\end{equation}
where $G$ is the $D-1\times D-1$ matrix defined by $G_{ij}:=\bra{\tilde{\psi}_i}C^{-1}(C^{-1})^\dagger\ket{\tilde{\psi}_j}$, $i,j=1,...,D-1$. The normalization of $\ket{\phi}$ implies that $\bra{d}G\ket{d}=1$. Defining $\ket{d'}:=G^{1/2}\ket{d}$, we have, by eqs. (\ref{gap}), (\ref{approx_gap}), that
\begin{equation}
\lambda_1\geq \tilde{\lambda}_1\min_{\braket{d'}{d'}=1}\bra{d'}G\ket{d'}.
\end{equation}
Finally, $G$ is a restriction to a $D-1$-dimensional subspace of the diagonal operator $C^{-1}(C^{-1})^\dagger$. It follows that the minimum eigenvalue of $G$ is greater than or equal to that of $C^{-1}(C^{-1})^\dagger$. In turn, the latter equals $\min_i\frac{1}{|c_i|^2}$. Since $\sum_i |c_i|^2=1$ and $|c_i|\geq \frac{1}{f}$, it follows that $\max_i|c_i|^2\geq 1-\frac{D-1}{f}$. Thus, we arrive at
\begin{equation}
\lambda_1\geq \frac{1}{1-\frac{D-1}{f}}\tilde{\lambda}_1=:\xi(D,f).
\end{equation}

Now, define the operator $H:=\id-\frac{1}{\xi(D,f)}H'$. By all the above, it follows that $H$ has just one positive eigenvalue, namely, $1$, with corresponding eigenvector $\ket{\psi}$. It therefore satisfies the conditions $\proj{\psi}-H\geq 0$, $\tr(H\proj{\psi})=1$.

It remains to prove that $H$ can be expressed as a linear combination of $Z_k,X_k,Y_k$ with bounded coefficients.

First, note that the effect of acting on an operator of the form $\ket{Z_k}\bra{Z_{l}}$ with $(C^{-1})^\dagger$ on the left and $C^{-1}$ on the right is just to multiply it by $\frac{1}{c^*_kc_l}$. On the other hand,
\begin{equation}
\ket{Z_k}\bra{Z_{k+1}}=X_k+iY_k-(1+i)(Z_k+Z_{k+1}).
\end{equation}
It follows that $H'$ equals:
\begin{equation}
\frac{1}{2}\sum_{k=1}^{D-1}\frac{1}{|c_k|^2}Z_k+\frac{1}{|c_{k+1}|^2}Z_{k+1}-\frac{1}{c_k^*c_{k+1}}\left(X_k+iY_k-(1+i)(Z_k+Z_{k+1})\right)-\frac{1}{c_kc^*_{k+1}}\left(X_k-iY_k-(1-i)(Z_k+Z_{k+1})\right).
\end{equation}
Taking into account that $\id=\sum_{k=1}^DZ_k$, we have that the sum of the moduli of all the coefficients multiplying $X_k$ ($Y_k$), for some $k=1,...,D-1$, in the expansion of $H$ is upper bounded by $\frac{f(D-1)}{\xi(D,f)}$. Similarly, the sum of the moduli of all the coefficients multiplying $Z_k$ in the expansion of $H$ is upper bounded by
\begin{equation}
D+\frac{(D-1)(1+\sqrt{2})f}{\xi(D,f)}.
\end{equation}
\hfill $\square$

Thus, let $\ket{\psi_j}$ satisfy $\tr(\psi_j Z_k)\geq \frac{1}{f_j}$. We invoke the previous lemma, thus obtaining an operator $H$ of the form (\ref{hamiltonian}), with coefficients $\{\lambda_k\}_k,\{\mu_k\}_k,\{\nu_k\}_k$ bounded by eq. (\ref{error_control}). By eq. (\ref{lower_bound_fid_psi}), we thus have that
\begin{align}
&\tr(\psi_j^{(3)}\psi_j)\geq 1-\left(D+\frac{(D-1)(1+\sqrt{2})f_j}{2\xi(D,f_j)}\right)\delta_Z-\frac{f_j(D-1)}{2\xi(D,f_j)}\delta_X-\frac{f_j(D-1)}{2\xi(D,f_j)}\delta_Y\nonumber\\
&-\left(D+\frac{(D-1)(1+\sqrt{2})f_j}{\xi(D,f_j}+2\frac{f_j(D-1)}{\xi(D,f_j)}\right)\delta_o=:1-\frac{(\delta^\psi_j)^2}{4}.
\end{align}
By Remark \ref{trace_dist_remark}, we conclude that
\begin{equation}
\|\psi_j^{(3)}-\psi_j\|_1\leq \delta^\psi_j.
\label{def:delta_psi}
\end{equation}
\subsubsection{Proof of Proposition \ref{prop:usual_techniques}}
\label{proof_usual_techniques}
Let $\ket{v}=\sum_{k=1}^D v_k\ket{Z_k}$ be an eigenvector of $\tilde{H}$ in eq. (\ref{h_tilde}) with eigenvalue $\lambda$. Then, $\ket{v}$ satisfies
\begin{align}
 &\frac{1}{2}v_1-\frac{1}{2}v_2=\lambda v_1,\nonumber\\
 &v_k-\frac{1}{2}(v_{k-1}-v_{k+1})=\lambda v_k,k=2,...,D-1,\nonumber\\
 &\frac{1}{2}v_D-\frac{1}{2}v_{D-1}=\lambda v_D.
\end{align}
This expression can be simplified if we define $v_{0}:=v_1, v_{D+1}:=v_D$. Then, it holds that
\begin{equation}
v_k-\frac{1}{2}(v_{k-1}-v_{k+1})=\lambda v_k,k=1,...,D.
\end{equation}
Let us try the ansatz $v_k=ae^{i\alpha k}+be^{-i\alpha k}$. This ansatz satisfies the equation above with $\lambda=1-\cos(\alpha)$, so it remains to impose the boundary conditions $v_{0}=v_1, v_{D+1}=v_D$, which translate as
\begin{equation}
a(e^{i\alpha}-1)+b(e^{-i\alpha}-1)=0,a(e^{i(D+1)\alpha}-e^{iD\alpha})+b(e^{-i(D+1)\alpha}-e^{-iD\alpha})=0.
\end{equation}
We are interested in solutions $(a,b)$ of this system of equations different from the trivial one $a=b=0$. Thus, we demand the corresponding determinant to vanish, i.e.,
\begin{equation}
\left|\begin{array}{cc}e^{i\alpha}-1&e^{-i\alpha}-1\\
e^{i(D+1)\alpha}-e^{iD\alpha}& e^{-i(D+1)\alpha}-e^{-iD\alpha}\end{array}\right|=0.
\end{equation}
The determinant equals $4i\sin(D\alpha)(1-\cos(\alpha))$. This means that, for $\alpha=k\frac{\pi}{D}$, with $k\in \Z$, the vector $\ket{v}$ is non-zero. Thus, the numbers $\{1-\cos\left(\frac{\pi j}{D}\right):j=0,...,D-1\}$ belong to the spectrum of $\tilde{H}$. They are $D$ different numbers and hence they are the whole spectrum. In particular, taking $j=0$, we find a $0$ eigenvalue, with eigenvector $\ket{v}$ satisfying $v_k=a+b$, for $k=1,...,D$.

\section{Wigner property of pure quantum state ensembles in dimension $D=2$}
\label{app:wigner_2}

The goal of this section is to prove the following result.

\begin{lemma}
\label{lemma:wigner_dim_2}
Let $\{\vec{m}_k\}_{k=1}^N\subset \R^3$, $\{\vec{n}_k\}_{k=1}^N\subset \R^3$ be sets of vectors of norm smaller than or equal to $1$ such that
\begin{align}
&|\vec{m}_i\cdot\vec{m_j}-\vec{n}_i\cdot\vec{n_j}|\leq \delta, \forall i,j.
\label{approxi}
\end{align}
Then, there exists an orthogonal transformation $\tilde{O}$ such that
\begin{align}
\|\vec{n}_k-\tilde{O}\vec{m}_k\|_2\leq \sqrt{N\delta}+\sqrt{\delta'c'},k=1,...,N,
\end{align}
with 
\begin{equation}
c':=\sum_i\frac{\|\vec{c}^i\|_1^2}{\lambda_i},
\end{equation}
\begin{equation}
\delta':=\delta+N\delta+2\sqrt{N\delta},
\end{equation}
where $\vec{c}_1,\vec{c}_2,...,$ are the eigenvectors of the Gram matrix $\bar{G}_{ij}=\vec{m}_i\cdot\vec{m_j}$, with non-zero eigenvalues $\lambda_1,\lambda_2,...$.
\end{lemma}
The Wigner property in dimension $D=2$ follows by noting that any pair of qubit states $\alpha,\beta\in B(\C^2)$ can be expressed as $\alpha=\frac{1}{2}(\id+\vec{m}_\alpha\cdot\vec{\sigma})$, $\beta=\frac{1}{2}(\id+\vec{m}_\beta\cdot\vec{\sigma})$, with $\|\vec{m}_\alpha\|_2, \|\vec{m}_\beta\|_2\leq 1$, where $\vec{\sigma}=(\sigma_x, \sigma_y,\sigma_z)$ are the Pauli matrices. In this \emph{Bloch representation}, $\alpha$ is pure iff $\|\vec{m}_\alpha\|=1$. In addition, $\tr(\alpha\beta)=\frac{1}{2}(1+\vec{m}_\alpha\cdot\vec{m}_\beta)$, and $\|\alpha-\beta\|_1=\|\vec{m}_\alpha-\vec{m}_\beta\|_2$.

It follows that state overlap constraints of the form 
\begin{equation}
|\tr(\psi_i\psi_j)-\tr(\bar{\psi}_i\bar{\psi}_j)|\leq \delta,\forall i,j,    
\label{condis2}
\end{equation}
are equivalent to 
\begin{equation}
\|\vec{m}_i\cdot\vec{m}_j-\vec{n}_i\cdot\vec{n}_j\|_2\leq 2\delta,\forall i,j
\end{equation}
where $\{\vec{m}_i\}_i$ (resp. $\{\vec{n}_i\}_i$) denote the Bloch vectors of $\{\psi_i\}_i$ (resp. $\{\bar{\psi}_i\}_i$).

By Lemma \ref{lemma:wigner_dim_2}, condition (\ref{condis2}) thus implies that, for $\delta\ll 1$, there exists an orthonormal matrix $\tilde{O}$ such that $\vec{n}_k\approx \tilde{O}\vec{m}_k\forall k$ (in $2$-norm). Since any orthogonal transformation on Bloch vectors corresponds to a unitary or anti-unitary transformation in $B(\C^2)$, and the trace distance between two states coincides with the $2$-norm of the difference of their Bloch vectors, we conclude that arbitrary ensembles of pure two-dimensional states indeed satisfy the Wigner property.

It just remains to prove Lemma \ref{lemma:wigner_dim_2}. Before we proceed, we state two easily verifiable (and well-known) remarks.

\begin{remark}
\label{remark_norms}
Let $G, \bar{G}$ be two positive semidefinite matrices of the same size. Then,
\begin{equation}
\|\bar{G}^{1/2}-G^{1/2}\|\leq \|\bar{G}-G\|^{1/2}.
\end{equation}
\end{remark}
\begin{proof}
Let $\ket{\omega}$ be the normalized eigenvector of $\bar{G}^{1/2}-G^{1/2}$, with eigenvalue $\mu$, such that $|\mu|=\|\bar{G}^{1/2}-G^{1/2}\|$. Then, 
\begin{equation}
\bra{\omega}(\bar{G}-G)\ket{\omega}=\bra{\omega}[(\bar{G}^{1/2}-G^{1/2})\bar{G}^{1/2}+G^{1/2}(\bar{G}^{1/2}-G^{1/2})]\ket{\omega}=\mu\bra{\omega}(\bar{G}^{1/2}+G^{1/2})\ket{\omega}.    
\end{equation}
 We thus have that
 \begin{equation}
 |\bra{\omega}(\bar{G}^{1/2}-G^{1/2})\ket{\omega}|^2\leq |\bra{\omega}(\bar{G}^{1/2}-G^{1/2})\ket{\omega}|\bra{\omega}(\bar{G}^{1/2}+G^{1/2})\ket{\omega}=|\bra{\omega}(\bar{G}-G)\ket{\omega}|\leq \|\bar{G}-G\|.
 \end{equation}
\end{proof}
Note: this proof is not original, see the comment by user \texttt{jwelk} in \url{https://math.stackexchange.com/questions/1934184/is-the-matrix-square-root-uniformly-continuous}.

\begin{remark}
\label{remark_orthogonal}
Let $Z$ be a square complex (real) matrix such that $\|\id-Z^\dagger Z\|\leq \delta$. Then, there exists a unitary (orthogonal) matrix $O$ such that
\begin{equation}
\|Z-O\|\leq \sqrt{\delta}.
\end{equation}
\end{remark}
\begin{proof}
Let $Z=O\sqrt{Z^\dagger Z}$ be the polar decomposition of $Z$, with $O$ unitary or real, depending on whether $Z$ is complex or real. Then,
\begin{equation}
(Z-O)^\dagger(Z-O)=Z^\dagger Z+\id-2\sqrt{Z^\dagger Z}=(\id-\sqrt{Z^\dagger Z})^2.
\end{equation}
By relation $\|\id-Z^\dagger Z\|\leq \delta$ and Remark \ref{remark_norms}, the operator norm of the right-hand side of the equation is upper bounded by $\delta$. Thus, the norm of $Z-O$ is upper bounded by $\sqrt{\delta}$.    
\end{proof}

Now we can prove Lemma \ref{lemma:wigner_dim_2}. 

\begin{proof}
Define the Gram matrices $\bar{G}_{ij}:=\vec{m}_i\cdot\vec{m_j}$, $G_{ij}:=\vec{n}_i\cdot\vec{n_j}$. By relation (\ref{approxi}), we have that $\|\bar{G}-G\|\leq N\delta$. Next, we invoke Remark \ref{remark_norms}, concluding that $\|\bar{G}^{1/2}-G^{1/2}\|\leq \sqrt{N\delta}$. Note that the vectors $\{\bar{G}^{1/2}\ket{k}\}_k$, $\{\ket{m_k}\}_k$ have the same overlaps: it follows that there exists a real isometry $\bar{O}$ such that $\bar{G}\ket{k}=\bar{O}\ket{m_k}$, for all $k$. Similarly, there exists a real isometry $O$ with $G^{1/2}\ket{k}=\bar{O}\ket{n_k}$, for all $k$. On the other hand, 
\begin{equation}
\|\bar{G}^{1/2}\ket{k}-G^{1/2}\ket{k}\|\leq \|\bar{G}^{1/2}-G^{1/2}\|\leq \sqrt{N\delta}.
\end{equation}
We conclude that, for all $k$, there exists $\ket{v_k}$ with $\|\ket{v_k}\|\leq \sqrt{N\delta}$ such that
\begin{equation}
\bar{O}\ket{m_k}+\ket{v}_k=O\ket{n_k}.
\end{equation}
It thus follows that
\begin{equation}
O^T\bar{O}\ket{m_k}+\ket{\tilde{v}_k}=\ket{n_k},
\label{casi_ortogonal}
\end{equation}
with $\|\ket{\tilde{v}_k}\|\leq \sqrt{N\delta}$. 

It remains to prove that we can approximate $O^T\bar{O}$ by an orthogonal matrix. To this end, consider the subspace ${\cal H}\subset \R^3$ spanned by the vectors $\{\vec{m}_k\}_k$, and define ${\cal H}':=O^T\bar{O}{\cal H}$. We construct a new linear operator ${\cal O}:\R^3\to\R^3$ in the following way:
\begin{align}
&{\cal O}\rvert_{{\cal H}}:=O^T\bar{O}\rvert_{{\cal H}},\nonumber\\
&{\cal O}\rvert_{{\cal H}^\perp}:=S,
\label{def_O}
\end{align}
where $\bullet^\perp$ denotes the complement of $\bullet$, and $S$ is any isometry $S:{\cal H}^\perp\to({\cal H}')^\perp$.

To apply Remark \ref{remark_orthogonal} to ${\cal O}$, we need to study how ${\cal O}^\dagger{\cal O}$ differs from the identity map. By eq. (\ref{def_O}), it holds that 
\begin{equation}
{\cal O}^\dagger{\cal O}=(\bar{O}\rvert_{{\cal H}})^TOO^T\bar{O}\rvert_{{\cal H}}\oplus\id_{{\cal H}^\perp}.
\end{equation}
Thus, 
\begin{equation}
\|{\cal O}^\dagger{\cal O}-\id\|=\max_{\ket{w}}\bra{w}(\id-\bar{O}^TOO^T\bar{O})\ket{w},
\label{max_norm}
\end{equation}
where the maximum is over all normalized $\ket{w}\in {\cal H}$. 

Now, by eq. (\ref{casi_ortogonal}) we have that
\begin{equation}
|\braket{n_j}{n_k}-\bra{m_j}\bar{O}^TOO^T\bar{O}\ket{m_k}|\leq N\delta+2\sqrt{N\delta}.
\end{equation}
Invoking eq. (\ref{approxi}), we arrive at
\begin{equation}
\bra{m_j}(\id-\bar{O}^TOO^T\bar{O})\ket{m_k}\leq \delta+N\delta+2\sqrt{N\delta}=:\delta'.
\label{close_id}
\end{equation}
Define $d:=\mbox{dim}({\cal H})$, and let $\vec{c}^1,...,\vec{c}^{d}$ be the eigenvectors of $\bar{G}$ with respective non-zero eigenvalues $\lambda_1,...,\lambda_{d}$. Then, any normalized vector $\ket{w}\in {\cal H}$ can be expressed as
\begin{equation}
\ket{w}=\sum_{i,j}s_i\frac{c^i_j}{\sqrt{\lambda_i}}\ket{m_j},
\end{equation}
for some $\vec{s}\in\R^{d}$ with $\|\vec{s}\|_2=1$. By (\ref{close_id}), we thus have that
\begin{equation}
 \bra{w}(\id-\bar{O}^TOO^T\bar{O})\ket{w}\leq\delta'\left(\sum_{i}|s_i|\frac{\|\vec{c}^i\|_1}{\sqrt{\lambda_i}}\right)^2\leq \delta'\left(\sum_i\frac{\|\vec{c}^i\|_1^2}{\lambda_i}\right)=\delta'c'.
\end{equation}
Thus, by eq. (\ref{max_norm}),
\begin{equation}
\|{\cal O}^\dagger{\cal O}-\id\|\leq \delta'c'.
\end{equation}
From Remark \ref{remark_orthogonal}, we have that there exists a $3\times 3$ orthogonal matrix $\tilde{O}$ that approximates $\bar{O}$ up to an error $\sqrt{\delta'c'}$. By eq. (\ref{casi_ortogonal}) we then have that 
\begin{equation}
\|\tilde{O}\ket{m_k}-\ket{n_k}\|\leq \sqrt{N\delta}+\sqrt{\delta'c'}.
\end{equation}

\end{proof}

\section{A robust version of POVM detection}
\label{app:robust_POVM}
\begin{lemma}
\label{lemma_robust_POVMs}
Let $M=(M_b)_{b=1}^B$ be an extremal, $D$-dimensional POVM, such that, for $b=1,...,B$, the POVM element $M_b$ admits the following spectral decomposition: $M_b=\sum_i\lambda^i_b\proj{\xi_{i}^b}$, for $\lambda_b^i>0$. Let $Z_b=\sum_{i=1}^{d_B}\proj{\psi_i^b}$ be the projector onto the kernel of $M_b$, and let $\{\phi_i^b:i=1,...,d_B, b=1,...,B\}$ be a state ensemble (not necessarily pure) with the property that $\|\phi_i^b-\psi_i^b\|_1\leq \epsilon,\forall b, i$. Then, any POVM $\bar{M}$ satisfying
\begin{equation}
\sum_{b}\tr(Z_b\bar{M}_b)\leq \delta
\label{POVM_viol}
\end{equation}
is such that
\begin{equation}
\|M_b-\bar{M}_b\|\leq (1+ \sqrt{D}B\sqrt{\|G^{-1}\|})\epsilon',\forall b,
\end{equation}
where $\epsilon':=(D-1)\epsilon+\delta+2\sqrt{(D-1)\epsilon+\delta}$, and $G$ is the Gram matrix of the vectors $\{\ket{\xi_{i}^b}\bra{\xi_j^b}\}_{i,j,b}$.
\end{lemma}
To prove the lemma, we need two auxiliary propositions.
\begin{prop}\cite{zhang2006schur}
\label{prop:schur}
Let $X$ be a matrix of the form:
\begin{equation}
X=\left(\begin{array}{cc}A&B^\dagger\\B &C\end{array}\right),
\end{equation}
and let $V$ be any isometry from the support of $C$ to a space with the same dimensionality. Then, $X\geq 0$ iff $VCV^\dagger>0$, $V^\dagger VB=B$ and $A-B^\dagger V^\dagger (VCV^\dagger)^{-1}VB\geq 0$.
\end{prop}

\begin{prop}
\label{prop_grams}
Let $\{\ket{v_i}\}_i\in \C^D$ be a set of linearly independent vectors, with Gram matrix $G_{ij}=\braket{v_i}{v_j}$. Let $\ket{w}=\sum_i c_i\ket{v_i}$. Then,
\begin{equation}
\|\vec{c}\|_2\leq \sqrt{\|G^{-1}\|}\|\ket{w}\|_2.
\end{equation}
\end{prop}
\begin{proof}
Define $\Gamma=\sum_i\ket{i}\bra{v_i}$. Then, 
\begin{equation}
\bra{j}\Gamma\ket{w}=\braket{v_j}{w}=\sum_ic_i\braket{v_j}{v_i}=\bra{j}G\ket{c}.
\end{equation}
It follows that
\begin{equation}
\ket{c}=G^{-1}\Gamma\ket{w}.
\end{equation}
Thus, $\|\ket{c}\|_2\leq \|\ket{w}\|_2\|G^{-1}\Gamma\|$. Finally, the norm of $G^{-1}\Gamma$ is the square root of the norm of
\begin{equation}
G^{-1}\Gamma\Gamma^\dagger G^{-1}=G^{-1}G G^{-1}=G^{-1}.
\end{equation}

\end{proof}

Now we are ready to prove Lemma \ref{lemma_robust_POVMs}.

\begin{proof}
From (\ref{POVM_viol}), we have that
\begin{equation}
\tr\{Z_b \bar{M}_b\}\leq d_B\epsilon+\delta.
\end{equation}
Define $Z^\perp_b:=\id-Z_b$, and let $V$ be an isometry from the support of $Z_b \bar{M}_b Z_b$ to a space with the same dimensionality. Then the above implies that $\|VZ_b \bar{M}_b Z_bV^\dagger\|\leq d_B\epsilon+\delta$. On the other hand,
\begin{equation}
M_b=Z^\perp_b \bar{M}_b Z^\perp_b+Z^\perp_b\bar{M}_bZ_b+Z_b\bar{M}_bZ^\perp_b+Z_b\bar{M}_bZ_b.
\label{decomp_M}
\end{equation}
Since $M_b\geq0$, we have, by Proposition \ref{prop:schur}, that 
\begin{align}
&Z^\perp_b \bar{M}_bZ^\perp_b-\frac{1}{d_B\epsilon+\delta}Z^\perp_b\bar{M}_bZ_b\bar{M}_bZ^\perp_b\geq Z^\perp_b\bar{M}_bZ^\perp_b-\nonumber\\
&Z^\perp_b\bar{M}_bZ_bV^\dagger (VZ_b\bar{M}_bZ_bV^\dagger)^{-1}VZ_b\bar{M}_bZ^\perp_b \geq 0.
\end{align}

Thus, the norm of the second and third terms of the right-hand side of eq. (\ref{decomp_M}) is bounded by $\sqrt{d_B\epsilon+\delta}$. It follows that
\begin{align}
&\|\bar{M}_b-Z^\perp_b \bar{M}_b Z^\perp_b\|\leq d_B\epsilon+\delta+2\sqrt{d_B\epsilon+\delta} \nonumber\\
&\leq(D-1)\epsilon+\delta+2\sqrt{(D-1)\epsilon+\delta}=\epsilon'.
\end{align}
To bound $\|M_b-\bar{M}_b\|$ it thus remains to bound $\|M_b-\tilde{M}_b\|$, with $\tilde{M}_b=Z_b^\perp \bar{M}_bZ_b^\perp$. Consider then the chain of inequalities:
\begin{align}
&\|\sum_b M_b-\tilde{M}_b\|_2\leq \|\sum_b M_b-M_b\|_2+\|\sum_b \bar{M}_b-\tilde{M}_b\|_2\leq \nonumber\\
&\sqrt{D}\sum_b \|\bar{M}_b-\tilde{M}_b\|_\infty\leq \sqrt{D}B\epsilon',
\end{align}
where, to arrive at the third expression, we made use of the identity $\sum_b M_b-\bar{M}_b=\id-\id=0$, the triangular inequality and the norm relation $\|\bullet\|_2\leq \sqrt{D}\|\bullet\|$.

Now, let $\{\ket{\xi_{i}^b}\}_i$ be the eigenvectors of $M_b$ with non-zero eigenvalue, and let $G$ be the Gram matrix of the operators $\{\ket{\xi_{i}^b}\bra{\xi_j^b}\}_{i,j,b}$, with the usual scalar product $\langle \alpha,\beta\rangle:=\tr(\alpha^\dagger\beta)$. If $(M_b)_b$ is, indeed, extremal, then by \cite{POVMs_dariano} $G>0$, i.e., $G$ is invertible. Let $M_b-\tilde{M}_b=\sum_{i,j,b}c_{ij}^b\ket{\xi_i^b}\bra{\xi_j^b}$. By Proposition \ref{prop_grams} we then have that
\begin{align}
\|M_b-\tilde{M}_b\|&\leq \sqrt{\sum_d\|M_d-\tilde{M}_d\|^2_2}=\sqrt{\sum_{i,j,d}|c_{ij}^d|^2}\nonumber\\
&=\|\vec{c}\|\leq \sqrt{D}B\sqrt{\|G^{-1}\|}\epsilon'.
\end{align}
It follows that
\begin{align}
\|M_b-\bar{M}_b\|&\leq \|M_b-\tilde{M}_b\|+\|\bar{M}_b-\tilde{M}_b\| \nonumber\\
&\leq(1+ \sqrt{D}B\sqrt{\|G^{-1}\|})\epsilon'.
\end{align}

\end{proof}

With Lemma \ref{lemma_robust_POVMs}, it is easy to prove analytic robustness bounds on self-testing of POVMs.

\begin{lemma}
Let $M=(M_b)_{b=1}^B$ be an extremal, $D$-dimensional POVM, and let $Z_b=\sum_{i=1}^{d_B}\proj{\psi_i^b}$ be the projector onto the kernel of $M_b$. Consider the linear witness
\begin{equation}
W_M(P):=W_\phi(P)-\sum_{b,i}P(b|x=(i,b),y=Y),
\end{equation}
where $\phi$ is an ensemble of pure states containing $\{\psi_i^b:i,b\}$ (and also ${\cal T}$, if $D>2$) and satisfying (\ref{max_mixed_cond}) in the main text. The considered prepare-and-measure scenario has $X=|\phi|$ preparations and $Y=\frac{X^2-X}{2}+1$ measurements, the last of which, $y=Y$, has $B$ outcomes.

Suppose that some $D$-dimensional realization $(\bar{\phi},\bar{M})$ generates $P$ such that
\begin{equation}
W_M(P)\geq 1-\frac{1}{D}-\epsilon.
\label{viol_M}
\end{equation}
Then, 
\begin{equation}
\|M_{b|Y}-\bar{M}_{b|Y}\|\leq \delta',\forall b,
\end{equation}
with $\delta'\leq O(\epsilon^{1/16})$, for $D=2$ or $\delta'\leq O(\epsilon^{1/32})$, otherwise.
\end{lemma}

\begin{proof}
Eq. (\ref{viol_M}) implies that
\begin{align}
&W_\psi(P)\geq 1-\frac{1}{D}-\epsilon,\\
&\sum_{b,i}P(b|x=(i,b),y=Y)\geq \epsilon.\label{restr_POVM}
\end{align}
Thus, by the main text, the states $\{\bar{\psi}_{i,b}:i,b\}$ are $\delta$-away from the reference states $\{\psi_i^b:i,b\}$, with $\delta=O(\epsilon^{1/8})$, for $D=2$, or $\delta=O(\epsilon^{1/16})$, otherwise. With (\ref{restr_POVM}) and the extremality of $M$, the conditions of Lemma \ref{lemma_robust_POVMs} are thus satisfied, resulting in the statement of the present lemma.
    
\end{proof}

\section{Examples of pure state ensembles that do not satisfy Wigner's property}
\label{app:SIC}

In this section, we present various families of pure state ensembles that do not satisfy the Wigner property. The first family is quite generic, but none of its members satisfies condition (\ref{max_mixed_cond}) in the main text. Thus it is a possibility that the minimum completion of each ensemble such that condition (\ref{max_mixed_cond}) in the main text holds is enough to self-test the ensemble via the linear witness in Lemma \ref{witness_norms} in the main text. The subsequent examples in this appendix consist of sets of nine pure qutrit states forming SIC POVMs. These ensembles satisfy condition (\ref{max_mixed_cond}) in the main text, and so their overlaps can be self-tested without the need of more state preparations. That would not be enough to self-test the quantum states, though.

\subsection{Pure qubit state ensembles embedded in the three-dimensional space}
\label{wig1}

It has been shown in the main text that all pure qubit state ensembles $\{\psi_i\}_{i=1}^N\subset P$ satisfy the Wigner property. However, it is important to note that this statement holds true only when $\mathcal{P}$ in Definition \ref{def_wigner} in the main text is the set of all rank-1 projectors in the two-dimensional space. In other words, $\{\phi_i\}_{i=1}^N$ are restricted to be qubits as well. The above statement does not hold without this restriction, that is if $\psi_i$ are in a qubit space embedded in a higher dimensional space. To illustrate this, let us consider the following three-element set:
\begin{equation}
\{|\psi_1\rangle,|\psi_2\rangle,|\psi_3\rangle\}=\left\{|1\rangle,b_{21}|1\rangle+b_{22}|2\rangle,b_{31}|1\rangle+b_{32}e^{i\beta}|2\rangle\right\},
\label{eq:qpsiset}
\end{equation}
where $b_{21},b_{22},b_{31}$ and $b_{32}$ are real non-negative numbers and $b_{21}^2+b_{22}^2=b_{31}^2+b_{32}^2=1$. It is worth noting that any set of three normalized qubit state vectors can be transformed into this form by applying appropriate unitary transformations and multiplying the vectors with global phases. Now consider the following set of vectors
\begin{equation}
\{|\phi_1\rangle,|\phi_2\rangle,|\phi_3\rangle\}=\left\{|1\rangle,b_{21}|1\rangle+b_{22}|2\rangle,b_{31}|1\rangle+t_{32}|2\rangle+t_{33}|3\rangle\right\},
\label{eq:tphiset}
\end{equation}
where $t_{32}$ is given by the equation
\begin{equation}
b_{22}t_{32}=|b_{21}b_{31}+b_{22}b_{32}e^{i\beta}|-b_{21}b_{31},
\label{eq:t22}
\end{equation}
and $t_{33}=\sqrt{1-b_{31}^2-t_{32}^2}$.
It can be observed that $\mathrm{Tr}(\phi_i\phi_j)=\mathrm{Tr}(\psi_i\psi_j)$ (or equivalently $|\langle\phi_i|\phi_j\rangle|=|\langle\psi_i|\psi_j\rangle|$) holds true for all pairs of $i$ and $j$. For $i=j=3$ it is due to the above equation, while for other pairs of $i$ and $j$, it is trivially satisfied.

By considering that 
\begin{equation}
-b_{22}b_{32}\leq|b_{21}b_{31}-b_{22}b_{32}|-b_{21}b_{31}\leq b_{22}t_{32}\leq b_{22}b_{32}, 
\end{equation}
it becomes evident that $t_{33}$ is a real number. Notably, except for very special cases, $|\phi_i\rangle$ span the three-dimensional space, while the vectors $|\psi_i\rangle$ reside in a two-dimensional subspace. Therefore the projectors corresponding to these sets cannot be transformed into each other using either unitary or anti-unitary operations. Hence, if we allow the qubit space to be a subspace embedded within a higher-dimensional space, sets of $D=2$ vectors do not necessarily satisfy the Wigner property.

\subsection{Pure qutrit state ensembles forming SIC POVMs}
\label{wig2}

Let us consider a paradigmatic example of a non-projective POVM, known as a symmetric informationally complete (SIC) POVM, which consists of rank-one POVM elements~\cite{renes2004}. A SIC POVM in $D$ dimensions has $D^2$ outcomes (it is still an open question whether they exist for all $D$). The elements corresponding to the outcomes are proportional to one-dimensional projectors with equal proportionality factors. Moreover, for all pairs of projectors, the traces of their products are also equal. Since the projectors define a POVM, they satisfy condition (\ref{max_mixed_cond}) in the main text. 

Most known SIC POVMs have been constructed by finding a suitable fiducial vector and generating the other vectors corresponding to the projectors by repeatedly acting on the fiducial vector with two operators: the shift operator $X=\sum_{k=1}^D|k\pmod{D}+1\rangle\langle k|$ and the clock operator $Z=\sum_{k=1}^D e^{i2\pi(k-1)/D}|k\rangle\langle k|$. The following one-parameter three-dimensional set can be derived in this way:
\begin{align}
&\{|\psi_1(t)\rangle,|\psi_2(t)\rangle,|\psi_3(t)\rangle,|\psi_4(t)\rangle,|\psi_5(t)\rangle,|\psi_6(t)\rangle,|\psi_7(t)\rangle,|\psi_8(t)\rangle,|\psi_9(t)\}\rangle=\nonumber\\
&\bigg\{\frac{|1\rangle+e^{it}|2\rangle}{\sqrt{2}},\frac{|2\rangle+e^{it}|3\rangle}{\sqrt{2}},\frac{|3\rangle+e^{it}|1\rangle}{\sqrt{2}},\nonumber\\
&\frac{|1\rangle+e^{it+2\pi i/3}|2\rangle}{\sqrt{2}},\frac{|2\rangle+e^{it+2\pi i/3}|3\rangle}{\sqrt{2}},\frac{|3\rangle+e^{it+2\pi i/3}|1\rangle}{\sqrt{2}},\nonumber\\
&\frac{|1\rangle+e^{it-2\pi i/3}|2\rangle}{\sqrt{2}},\frac{|2\rangle+e^{it-2\pi i/3}|3\rangle}{\sqrt{2}},\frac{|3\rangle+e^{it-2\pi i/3}|1\rangle}{\sqrt{2}}
\bigg\},
\label{eq:3dsic}
\end{align}
where the fiducial vector is the first. It is not difficult to check that $\sum_{k=1}^{9}\psi_k/3=\id_3$, and that $\Tr(\psi_i\psi_j)=|\langle\psi_i|\psi_j\rangle|^2=1/2$ for any pair, so the above vectors define a SIC POVM in $D=3$ for any value of the parameter $t$, indeed. At the same time, they do not satisfy the Wigner property. We will show this firstly for the specific pair of sets corresponding to $t=0$ and $t'=\pi$. Then we will show it for other generic choices of $t$ and $t'$.

The traces of the triple products $\psi_i(t)\psi_j(t)\psi_k(t)$ are trivially elements of a finite set of discrete numbers. In fact, it is not difficult to calculate that these come from the following list of numbers
\begin{equation}
\label{phaselist}
\left[\frac{e^{\pm 3it}}{8},\:\: \frac{e^{(\pm 3it\pm 2\pi i/3)}}{8},\:\: \frac{e^{\pm\pi i/3}}{8},\:\: -\frac{1}{8}\right].
\end{equation}
Since all the absolute values are $1/8$, we are only interested in the phase of the traces,
\begin{equation}
\phi=\text{arg}(\Tr{\psi_i(t)\psi_j(t)\psi_k(t)}),
\end{equation}
where $-\pi\le\phi<\pi$. It can be seen that the odd permutations of the indices $i,j,k$ or the complex conjugation of the three operators in the trace invert the sign of the phase. So we take the absolute value of the phase, $|\phi|$, where we have $0\le |\phi|<\pi$. This value of $|\phi|$ no longer depends on the permutations of the indices $i,j,k$ and the complex conjugation of the three operators in the trace. In fact, it is a discrete version of the gauge invariant geometric phase in the projective Hilbert space~\cite{aharonov1987,mukunda1993}. From each of the SIC POVMs corresponding to $t=0$ and $t'=\pi$, we can obtain a total of $\binom{9}{3}=84$ values for $|\phi|$. Then the two SIC POVMs are considered (non)unitarily inequivalent if their $|\phi|$ angles have different patterns. However, one can see that two of the numbers, $|\phi|=0$ and $|\phi|=\pi$, never appear among the phases of the list~(\ref{phaselist}) for the case $t=\pi$, while they do for $t=0$, so the two lists are not the same after all. Thus, the two SIC POVMs, the one corresponding to $t=0$ and the other corresponding to $t'=\pi$, are indeed not related by a unitary or anti-unitary transformation. 

Let us now consider the generic case of $t$ and $t'$. It is easy to verify that 
\begin{equation}
\mathrm{Tr}(\psi_1(t)\psi_2(t)\psi_3(t))=\langle\psi_1(t)|\psi_2(t)\rangle\langle\psi_2(t)|\psi_3(t)\rangle\langle\psi_3(t)|\psi_1(t)\rangle=e^{-3it}/8.    
\end{equation}
This number is invariant under a unitary operation, while the result of an anti-unitary one is a complex conjugation. Hence, if $3(t-t')$ is not an integer multiple of $\pi$, then the first three members of the sets belonging to the parameters $t$ and $t'$ cannot be transformed into each other by a unitary or an anti-unitary operation. However, due to the symmetry of the sets, $\mathrm{Tr}(\psi_i(t)\psi_j(t))=\mathrm{Tr}(\phi_i\phi_j)$ holds not only if the set $\{\phi_i\}$ is equal to $\{\psi_i(t')\}$, but also if it is any permutation of $\{\psi_i(t')\}$. Therefore, we will now prove that there exists $t'$ such that $\{\psi_i(t)\}$ cannot be transformed into any permutation of $\{\psi_i(t')\}$, or equivalently, no permutation of $\{\psi_i(t)\}$ can be transformed into $\{\psi_i(t')\}$ with a unitary or anti-unitary operation. From the list of the traces of the triple products given in (\ref{phaselist}) it can be seen that there exists $t'$ such that $e^{3it'}$ differs from all of them. Then there are no three elements of $\{\psi_i(t)\}$ which can be transformed into the first three elements of $\{\psi_i(t')\}$ by unitary or anti-unitary operation, and consequently no permutation of $\{\psi_i(t)\}$ can be transformed into $\{\psi_i(t')\}$ this way. This concludes the proof that set $\{\psi_i(t)\}$ does not satisfy the Wigner property for any value of the parameter $t$.

\section{Example of self-testing of a non-projective qutrit POVM}
\label{app:examplePOVM}


Below we particularize the construction in the main text to self-test an extremal non-projective, non-rank-one POVM. Note that every two-outcome non-projective POVM is non-extremal. Thus, in dimension $D=3$, the simplest example of an extremal non-projective, non-rank-one POVM requires three outcomes, with ranks $(1,1,2)$.

We define the following states
\begin{align}
\ket{\phi_1} &=  (\sqrt{3}\ket{0} - \ket{1})/2,\nonumber\\
\ket{\phi_2} &=  (\sqrt{3}\ket{0} + \ket{1})/2,\nonumber\\
\ket{\phi_3} &= \ket{1},\nonumber\\
\ket{\phi_4} &= \ket{2},
\end{align}
which make up the following POVM elements:
\begin{align}
M_1 &= 2\phi_1/3,\nonumber\\
M_2 &= 2\phi_2/3,\nonumber\\
M_3 &= 2\phi_3/3+\phi_4.
\label{povm_example}
\end{align}
It is straightforward to verify that these matrices indeed define a POVM with the desired ranks. This POVM is extremal, as one can verify through the criterion in \cite{extreme_POVMs}. 

Following the procedure in the main text, we first define the set of necessary preparations. Some of them are defined by the basis vectors of the null space of the $M_b$ matrices:
\begin{align}
\ket{\psi_1^1}=&(\ket{0}+\sqrt{3}\ket{1})/2,\nonumber\\
\ket{\psi_2^1}=&\ket{2},\nonumber\\
\ket{\psi_1^2}=&(-\ket{0}+\sqrt{3}\ket{1})/2,\nonumber\\
\ket{\psi_2^2}=&\ket{2},\nonumber\\
\ket{\psi_1^3} =&\ket{0},
\label{setofstates}
\end{align}
which define the spectral decomposition $Z_b=\sum_i\ket{\psi_i^b}\bra{\psi_i^b}$. By construction, we have $\tr(Z_b M_b)=0$ for $b=1,2,3$. 

Starting from the above pure qutrit states, we define the pure state ensemble $\psi=\{\psi_i^b: i,b\}\cup{\cal T}\cup{\cal R}$, where ${\cal R}$ is chosen so that the overall ensemble $\psi$ satisfies the condition~(\ref{max_mixed_cond}) in the main text for some positive $(\alpha_i)_i$. To self-test $M$, we thus require a prepare-and-measure scenario with $X=5+(5\times 3-6)+2=16$ preparations and $Y=(16^2-16)/2=120$ dichotomic measurements, plus a $3$-outcome measurement. $M$ is then self-tested whenever the prepare-and-measure witness labeled (\ref{witness_POVMs}) in the main text is maximized.


\end{appendix}

\end{document}